\documentclass[fontsize=12pt,a4paper,headings=normal,
twoside=false,leqno,parskip=half-,abstract=true]{scrartcl}
\usepackage[english]{babel}
\usepackage{bm}
\usepackage[utf8]{inputenc}
\usepackage{hyperref}
\hypersetup{
 pdftitle={Attractor Theorems in HL},
 pdfauthor={Phillipo Lappicy},
 colorlinks=true,
 linkcolor=blue,
 citecolor=blue,
 filecolor=blue,
 urlcolor=blue}

\usepackage{csquotes}
\usepackage{float}

\usepackage{graphicx}
\usepackage[format=plain,labelfont=bf,font=small]{caption}
\usepackage{subfigure}
\usepackage{xcolor}
\usepackage[arrow, matrix, curve]{xy}
\usepackage{float}

\usepackage{caption}
\usepackage{comment}
\usepackage[textwidth=2cm,textsize=small,backgroundcolor=none]{todonotes}

\usepackage{tabulary}
\usepackage{array}

\usepackage{amsmath,amsthm}
\usepackage{amssymb} 
\usepackage{latexsym}
\usepackage{mathtools}

\usepackage{tikz}
\usetikzlibrary{arrows, arrows.meta, calc, intersections, decorations.markings, shapes}

\usepackage{pgfplots, xcolor}
\pgfplotsset{compat=1.18}
\usepgfplotslibrary{colormaps}

\usepackage[format=plain,labelfont=bf,font=small]{caption}
\usepackage{float}

\usepackage{wasysym}

\newtheorem{thm}{Theorem}[section]
\newtheorem{lem}[thm]{Lemma}
\newtheorem{prop}[thm]{Proposition}

\newenvironment{pf}[1][Proof]{\begin{trivlist}
\item[\hskip \labelsep {\bfseries #1}]}{\end{trivlist}}

\usepackage[notref,notcite,color,final 
]{showkeys}

\DeclareOldFontCommand{\bf}{\normalfont\bfseries}{\mathbf}
\DeclareOldFontCommand{\it}{\normalfont\itshape}{\textit}

\newcommand{\RR}{\mathbb R}

\newcommand{\BIIa}{\mathcal{B}_{N_1}}
\newcommand{\BIIb}{\mathcal{B}_{N_2}}
\newcommand{\BIIc}{\mathcal{B}_{N_3}}

\newcommand{\BVIIa}{\mathcal{B}_{N_1N_2}}
\newcommand{\BVIIb}{\mathcal{B}_{N_1N_3}}
\newcommand{\BVIIc}{\mathcal{B}_{N_2N_3}}

\newcommand{\KC}{\mathrm{K}^{\ocircle}}

\renewcommand{\S}{\cal S}

\newcommand{\Q}{{\bf \mathrm{Q}}}
\newcommand{\T}{{ \mathrm{T}}}

\begin{document}

\title{
\LARGE{The Bianchi IX Attractor in Modified Gravity}
}

\author{
 \\
{~}\\
Ester Beatriz*, Everaldo Bonotto* and Phillipo Lappicy**\\
\vspace{2cm}}

\date{ }
\maketitle
\thispagestyle{empty}

\vfill

$\ast$\\
ICMC, Universidade de S\~ao Paulo, Brazil\\

$\ast$ $\ast$\\
Universidad Complutense de Madrid and ICMAT, Spain\\


\newpage
\pagestyle{plain}
\pagenumbering{arabic}
\setcounter{page}{1}

\begin{abstract}
\noindent
We consider vacuum anisotropic spatially homogeneous models in certain modified gravity theories (such as Ho{\v{r}}ava--Lifshitz, $\lambda$--$R$ or $f(R)$ gravity), which are expected to describe generic spacelike singularities for these theories. These models perturb the well-known Bianchi models in general relativity (GR) by a parameter $v\in (0,1)$ with GR recovered at $v=1/2$. 
We prove an analogue of the well-known Ringström attractor theorem in GR to the supercritical theories: for any $v\in (1/2,1)$, \emph{all} solutions of Bianchi type~$\mathrm{IX}$ converge to an analogue of the Mixmaster attractor, consisting of Bianchi type I solutions (Kasner states) and heteroclinic chains of Bianchi type II solutions.
In contrast to GR, there are no solutions that converge to a different set other than the Mixmaster (such as the locally rotationally symmetric solutions in GR).

\medskip \noindent \textbf{Keywords:} Ordinary Differential Equations; Dynamical Systems; Global Attractor; Cosmology; Spacelike Singularity; Modified Gravity.

\medskip \noindent \textbf{MSC 2020:} 34D45, 37G35, 37N20, 83C20, 83C75, 83D05.  
\end{abstract}


\section{Introduction}\label{sec:intro}

\numberwithin{equation}{section}
\numberwithin{figure}{section}
\numberwithin{table}{section}



The Belinski, Khalatnikov and Lifshitz (BKL) picture suggests that generic singularities in general relativity (GR) are: (i) \textit{vacuum dominated}, (ii) \textit{local} and (iii) \textit{oscillatory}.
In this regard, vacuum spatially homogeneous anisotropic cosmologies, the Bianchi models, are expected to play a key role in the dynamical asymptotic behaviour, see \cite{mis69a,bkl70,bkl82,Mixmaster,ugg13a,ugg13b,Coley}.
%
%
Similarly, the vacuum anisotropic spatially homogeneous Bianchi models derived in \cite{HellLappicyUggla} are expected to describe generic singularities in certain modified higher-order gravity theories, such as Ho{\v{r}}ava-Lifshitz,  $\lambda$-$R$ or $f(R)$-gravity, see \cite{hor09a,hor09b,Bakas10,Miso11,Kamenshchik,HL_status_report,fR,fR2,MG}. 
Indeed, it is heuristically argued that there is an `asymptotically dominant' curvature term in Hilbert--Einstein action toward the initial singularity, yielding a respective dominant Bianchi model.
%
More importantly, it is expected that such models provide new insights in GR.

The dominant vacuum Bianchi models in Ho{\v{r}}ava gravity are given by
\begin{subequations}\label{intro_dynsyslambdaR}
\begin{align}
\Sigma_\alpha^\prime &= 4v(1-\Sigma^2)\Sigma_\alpha + {\cal S}_\alpha,\\
N_\alpha^\prime &= -2(2v\Sigma^2 + \Sigma_\alpha)N_\alpha, \label{intro_dynsyslambdaR_N}
\end{align}
%
for $\alpha = 1,2,3$, with parameter $v\in (0,1)$, and the constraints,
\begin{align}
0 &= 1 - \Sigma^2 - \Omega_k, \label{intro_cons1}\\
0 &= \Sigma_1 + \Sigma_2 + \Sigma_3,\label{intro_cons2}
\end{align}
\end{subequations}
where
\begin{subequations}
\begin{align}
\Sigma^2 &:= \frac16\left(\Sigma_1^2 + \Sigma_2^2 + \Sigma_3^2\right),\label{Sigma}\\
\Omega_k &:= N_1^2 + N_2^2 + N_3^2 - 2N_1N_2 - 2N_2N_3 - 2N_3N_1, \label{Omega_k}\\
{\cal S}_\alpha &:= -4[(N_\beta - N_\gamma)^2 - N_\alpha(2N_\alpha - N_\beta - N_\gamma)].
\end{align}
\end{subequations}
Here $(\alpha\beta\gamma)$ is a permutation of $(123)$. Note that $(.){}^\prime = d/d\tau$ denotes
the derivative with respect to the chosen time variable, $\tau_-$, which is in the opposite direction of
physical time. Hence, for expanding models, $\tau_-\rightarrow\infty$
describes the dynamics toward the initial singularity.


The Bianchi type of a solution of system \eqref{intro_dynsyslambdaR} is determined by the signs of the $N$-variables, which are invariant under the flow. This provides a stratification of phase-space by invariant sets of increasing dimension.
Bianchi type $\mathrm{I}$ consists of all $N_\alpha$ being zero (constituting a circle of equilibria), Bianchi type $\mathrm{II}$ consists of a single non-zero $N_\alpha$ (yielding a hemisphere containing heteroclinic orbits between the equilibria of type $\mathrm{I}$), Bianchi types $\mathrm{VI}_0,\mathrm{VII}_0$ consist of two non-zero $N_\alpha$, and Bianchi types $\mathrm{VIII},\mathrm{IX}$ consist of three non-zero $N_\alpha$.
%
See Table~\ref{intro:classAmodels}.

\begin{table}[H]
\begin{center}
\begin{tabular}{|c|ccc|c|}
\hline Bianchi type & ${N}_\alpha$ &  ${N}_\beta$ & ${N}_\gamma$ & Dimension \\ \hline
$\mathrm{IX}$ & $+$& $+$& $+$ & 4\\
$\mathrm{VIII}$ & $-$& $+$& $+$ & 4\\
$\mathrm{VII}_0$ & $0$& $+$ & $+$ & 3\\
$\mathrm{VI}_0$ & $0$ & $-$ & $+$ & 3\\
II & $0$ & $0$ & $+$ & 2\\
I & $0$ & $0$ & $0$ & 1\\\hline
\end{tabular}
\caption{The Bianchi types are invariant subsets of phase-space characterized by the zeros and signs of the variables $(N_\alpha, N_\beta, N_\gamma)$,
where $(\alpha\beta\gamma)$ is a permutation of $(123)$. 
There are equivalent representations associated with an overall change of sign. 
} \label{intro:classAmodels}
\end{center}
\end{table}\vspace{-0.5cm}
Equations \eqref{intro_dynsyslambdaR} describe vacuum spatially homogeneous models in GR when $v=1/2$, and therefore $v\in (0,1)$ denotes perturbations of GR.
The dynamics of the ODEs \eqref{intro_dynsyslambdaR} in GR, when $v=1/2$, has been extensively considered in the literature. A major achievement is the attractor theorem by H. Ringström, which states that the $\omega$-limit set of generic solutions of Bianchi type IX is contained in the space of solutions of Bianchi types I and II, see \cite{Ringstrom, heiugg09b}.
Therefore, it is expected that the concatenation of heteroclinic orbits of Bianchi type II and the induced map of Bianchi type I (the so-called \textit{BKL map}) play a key role in the dynamics.
Stability of some concatenations can be found in \cite{Beguin,Liebscher,lieetal12,Bernhard,Dutilleul}. For an overview, see \cite{Mixmaster}.

In the case that $v\neq 1/2$, some of the features of GR persist.
In particular, the Bianchi hierarchy of invariant sets remains valid depending on whether the variables $N_\alpha$ are nonzero. Bianchi type $\mathrm{I}$  consists of all $N_\alpha$ being zero, Bianchi type II consists of a single non-zero $N_\alpha$, Bianchi types $\mathrm{VI}_0,\mathrm{VII}_0$ consist of two non-zero $N_\alpha$, and Bianchi types $\mathrm{VIII},\mathrm{IX}$ consist of three non-zero $N_\alpha$.
The authors in \cite{HellLappicyUggla} describe the dynamics of the types $\mathrm{I},\mathrm{II},\mathrm{VI}_0,\mathrm{VII}_0$ for all $v\in (0,1)$. In particular, it is shown that in GR, $v=1/2$, is a bifurcation where chaos becomes generic; see also \cite{LappicyDaniel}. A counter-example of the BKL conjecture was found when $v=0$, where there are periodic orbits that are far from the Mixmaster attractor, see \cite{LappicyLessard}.

We prove that in the Bianchi type IX supercritical case, $v \in (1/2,1)$, the global attractor, ${\cal A}_-$, consists of Bianchi types $\mathrm{I}$ and $\mathrm{II}$. 
Our main result is stated as follows. 
\begin{thm}
\label{mainthm}
For the supercritical case, $v \in (1/2,1)$,
all Bianchi type IX solutions 
converge (as $\tau_-\rightarrow \infty$) to the Bianchi types $\mathrm{I}$ and $\mathrm{II}$ subsets.
\end{thm}
We emphasize that there is no meagre set that does not converge to Bianchi types $\mathrm{I}$ and $\mathrm{II}$, as the locally rotationally symmetric set when $v=1/2$.
Note that we do not settle \cite[Conjecture 7.4]{HellLappicyUggla} that claims that there is a stable set $S$ of Bianchi type $\mathrm{I}$ solutions which is the generic attractor.
Indeed, we expect that generic solutions do not oscillate following types $\mathrm{I}$ and $\mathrm{II}$ solutions, 
since infinite heteroclinic chains of types $\mathrm{I}$ and $\mathrm{II}$ that do not converge to the stable subset $S$ form to a Cantor set, which should only possess a codimension 1 (non-generic) stable manifold; see~\cite{HellLappicyUggla}. 
However, we cannot prove such a result using the current methods. 

\newpage

The remainder of this paper is organized as follows. 
Section \ref{sec:hier} explicitly describes the hierarchy of invariant sets of the Bianchi models by their increasing dimension. 
Section \ref{sec:conc} presents concluding remarks and open problems.


\section{Bianchi Hierarchy}\label{sec:hier}

We describe the stratification of invariant sets by increasing dimension; see Table~\ref{intro:classAmodels}.
We rewrite~\eqref{intro_dynsyslambdaR} using Misner variables
	$\Sigma_1 = -2  \Sigma_+,
	\Sigma_2 =  \Sigma_+ + \sqrt{3} \Sigma_-,
	\Sigma_3 =  \Sigma_+ - \sqrt{3} \Sigma_-$:
\begin{subequations}\label{full:subs}
	\begin{align}
	\Sigma_+^\prime  &= 4v(1-\Sigma^2)\Sigma_+ + \, {\cal S}_+, \label{full:Sigma+}\\
	\Sigma_-^\prime   &= 4v(1-\Sigma^2)\Sigma_- + \, {\cal S}_-, \label{full:Sigma-}\\
	N_1^\prime 			  &= -2(2v\Sigma^2 - 2\Sigma_+) \, N_1, \label{full:N1} \\
	N_2^\prime 			  &= -2(2v\Sigma^2 + \Sigma_+ + \sqrt{3}\Sigma_-)N_2, \label{full:N2} \\
	N_3^\prime 			  &= -2(2v\Sigma^2 + \Sigma_+ - \sqrt{3}\Sigma_- )N_3, \label{full:N3} 
	\end{align}
\end{subequations}
bound to the following constraint:
\begin{equation} \label{constraint}
1 =  \Sigma^2 + N_1^2 + N_2^2 + N_3^2 - 2N_1N_2 - 2N_2N_3 - 2N_3N_1,
\end{equation}
which restricts the phase-space $\RR^5$ to a four-dimensional invariant submanifold, where 
\begin{subequations}
	\begin{align}
	\Sigma^2  &	:= \Sigma_+^2 + \Sigma_-^2 ,\label{Sigma2}\\
	\S_+ &  := 2 \left[ \left(N_3-N_2\right)^2 - N_1\left(2N_1-N_2-N_3\right) \right],\\
	\S_- &  := 2\sqrt{3} \left(N_3-N_2\right)\left(N_1-N_2-N_3\right).
	\end{align}
\end{subequations}

\subsection{Bianchi Type I}

Bianchi type I solutions are characterized by all $N_\alpha = 0,\alpha=1,2,3$.  The constraint \eqref{constraint} reduces the phase-space to the  \textit{Kasner circle} of equilibria:
\begin{align} \label{defKC}
\KC := \left\{ \left(\Sigma_+,\Sigma_-, 0,0,0\right) \in \RR^5 \,\, |  \,\,
\Sigma_+^2 + \Sigma_-^2 = 1 \right\}.
\end{align}
There are three special points in $\KC$ corresponding to the Taub representation of Minkowski spacetime in GR. They are therefore called the \emph{Taub points} and given by
\begin{align}\label{Taubpoints}
\T_1 := \left(-1,0\right), \hspace{1.2cm}
\T_2 := \left(\dfrac{1}{2},\frac{\sqrt 3}{2}\right), \hspace{1.2cm} \T_3 := \left(\dfrac{1}{2},-\dfrac{\sqrt 3}{2}\right).
\end{align}
Note the existence of the Kasner circle is independent on the parameter $v\in [0,1]$. Its stability, however, depends strongly on the parameter and this affects the type II solutions, which will be described next. 

In general, linearization of equation~\eqref{full:subs} at $\mathrm{K}^{\ocircle}$ results in $N'_1 = 2(\Sigma_+|_{\mathrm{K}^{\ocircle}}-v){N}_1$, and thereby the stability behaviour of $N_{1}$ changes when $\Sigma_+|_{\mathrm{K}^{\ocircle}} = v$. We define the \emph{unstable Kasner arc}, denoted by $\mathrm{int}(A_{1})$, to be the points in $\mathrm{K}^{\ocircle}$ that are unstable in the $N_1$ variable, i.e., when $\Sigma_+ > v$. The closure of $\mathrm{int}(A_{1})$ is denoted by
$A_1$ and is given by
\begin{equation}\label{A_1}
A_1:= \left\{ (\Sigma_+,\Sigma_-,0,0,0)\in \mathrm{K}^{\ocircle} \text{ $|$ }
\Sigma_+ \geq v\right\}.
\end{equation}
Note that the arcs $A_1$ are symmetric portions of
$\mathrm{K}^{\ocircle}$ with points $\mathrm{Q}_1 := -\mathrm{T}_1$ in the middle. 
The boundary set $\partial A_1$ consists of two fixed points, which coincide with the Taub points $\mathrm{T}_2$ and $\mathrm{T}_3$ when $v=1/2$, but unfolds each Taub point into two non-hyperbolic fixed points when $v\neq 1/2$.
Equivariance yields the arcs $A_2,A_3$, where the respective variables $N_2,N_3$ are unstable, and leads to similar statements for $\mathrm{T}_2$ and $\mathrm{T}_3$.
%
Note that the union of the three arcs $A_\alpha$ do not cover $\mathrm{K}^{\ocircle}$. For $v>1/2$, there is a closed region around the Taub points $\mathrm{T}_\alpha$ which is stable, defined by
\begin{equation}\label{setS}
    S:= \mathrm{K}^{\ocircle} \backslash \mathrm{int}(A_1) \cup \mathrm{int}(A_2) \cup \mathrm{int}(A_3).
\end{equation}
Moreover, the set $S$ consists of a curve of fixed points, each of which has a three-dimensional stable manifold that foliates its local neighborhood.

\subsection{Bianchi Type II}

Bianchi type II solutions consist of three disjoint hemispheres with a common boundary: the Kasner circle.
More precisely, it is the set of solutions where two $N$-variables are zero and one is nonzero, i.e. $\BIIa \cup \BIIb \cup \BIIc$, where  
\begin{align}\label{BII_N_1}
\BIIa :=\left\{ \left( \Sigma_+,\Sigma_-, N_1,0,0\right) \in \RR^5 \,\, \Big|  \begin{array}{c}
N_1 > 0  \\
{N_1}^2 = 1-\Sigma^2
\end{array} \right\},
\end{align}
and $\BIIb, \BIIc$ are obtained by equivariance with non-zero $N_2,N_3$, respectively.
Solutions of \eqref{full:subs} in the hemisphere $\BIIa$ are heteroclinic orbits between two Kasner equilibria with $\alpha$-limit sets in $\text{int}(A_1)$ and $\omega$-limit sets in the complement $A_1^c:=\mathrm{K}^{\ocircle} \backslash A_1$, as depicted in Figure \ref{KC_3d_map}. Similarly for $\BIIb, \BIIc$. For more details, see \cite{HellLappicyUggla}.
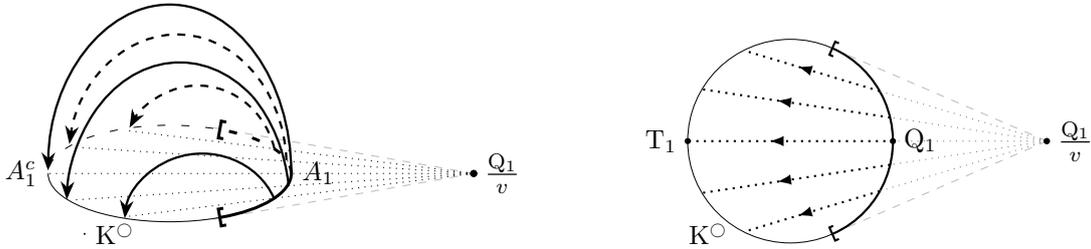
\begin{figure}[H]
	\centering
	\begin{tikzpicture}[scale=1.6]

 	\draw[lightgray,dashed,-] (2.5,0) -- (0.3,0.39);
 	\draw[lightgray,dashed,-] (2.5,0) -- (0.3,-0.39);

 	\draw [line width = 0.1mm] (-1,0) arc (180:360:1cm and 0.4cm);
 	\draw [line width = 0.1mm, loosely dashed] (1,0) arc (0:180:1cm and 0.4cm);
 	
 	\draw [-Bracket, line width = 0.4mm] (1,0) arc (0:-66.4218:1cm and 0.4cm);
 	\draw [-Bracket, line width = 0.4mm, loosely dashed] (1,0) arc (0:66.4218:1cm and 0.4cm);
 	\draw [black] (1,0) circle (0.01pt) node[anchor=west] {\footnotesize $A_1$};
  	\draw [black] (-1,0) circle (0.01pt) node[anchor=east] {\footnotesize $A_1^c$};

 	\filldraw [black] (-0.7,-0.5) circle (0.1pt) node[anchor=west] {\footnotesize $\KC$};

 	\draw [ -Stealth, line width = 0.3mm] (1,0) arc (0:179:1cm and 1.4cm);

	\draw [ -Stealth, line width = 0.3mm]  (0.86,-0.21) arc (22:174:0.64cm and 0.6cm);
	\draw [ -Stealth, line width = 0.3mm, dashed]  (0.86,0.21) arc (15:150:0.65cm and 0.7cm);

    \draw [line width = 0.3mm]  (0.68, 0.64) arc (48:8:0.93cm and 1.23cm);
    \draw [-Stealth,line width = 0.3mm]  (0.68, 0.64) arc (50:177:0.93cm and 1.2cm);
    \draw [-Stealth, line width = 0.3mm, dashed]  (0.96, 0.11) arc (3:170:0.90cm and 1.1cm);
 	
 	\draw[dotted] (2.5,0) -- (-1,0);
 	\draw[dotted] (2.5,0) -- (-0.8,0.22);
 	\draw[dotted] (2.5,0) -- (-0.8,-0.22);
	\draw[dotted] (2.5,0) -- (-0.35,0.356);
 	\draw[dotted] (2.5,0) -- (-0.35,-0.356);


%
%
%
%
%
%
%
%

 	\filldraw [black] (2.5,0) circle (0.7pt);
 	\draw [black] (2.5,0) circle (0.7pt) node[anchor=west] {$\frac{\Q_1}{v}$};
 	

	\end{tikzpicture}
	\qquad\quad
	\begin{tikzpicture}[scale=1.35, decoration={markings,mark=at position 0.6 with {\arrow{latex}}},>=latex]


	\draw[lightgray,dashed, -] (2.5,0) -- (0.41,-0.91);
	\draw[lightgray,dashed, -] (2.5,0) -- (0.41, 0.91);

	\draw [very thin] (1,0) arc (0:360:1cm and 1cm);	 	
	
	\draw [thick,Bracket-Bracket] (0.4,-0.9165) arc (-66.4218:66.4218:1cm and 1cm);
	
	\filldraw [black] (-1,0) circle (0.7pt) node[anchor=east] {\footnotesize $\T_1$};
	\filldraw [black] (1,0) circle (0.7pt) node[anchor=west] {\footnotesize $\Q_1$};
	
	\filldraw [black] (0.4,0.9165) circle (0pt);
	\filldraw [black] (0.4,-0.9165) circle (0pt);


\foreach \s in {0.15,0.34,0.5,0.66,0.85}
{\draw[gray,dotted] (2.5,0) -- ({cos(42-\s*84)},{sin(42-\s*84)});}
\draw[postaction={decorate},dotted,thick] ({cos(42-0.15*84)},{sin(42-0.15*84)}) -- ({cos(90+0.15*180)},{sin(90+0.15*180)});
\draw[postaction={decorate},dotted,thick] ({cos(42-0.5*84)},{sin(42-0.5*84)}) -- ({cos(90+0.5*180)},{sin(90+0.5*180)});
\draw[postaction={decorate},dotted,thick] ({cos(42-0.85*84)},{sin(42-0.85*84)}) -- ({cos(90+0.85*180)},{sin(90+0.85*180)});

\draw[postaction={decorate},dotted,thick] ({cos(42-0.34*84)},{sin(42-0.34*84)}) -- ({cos(90+0.33*180)},{sin(90+0.33*180)});
\draw[postaction={decorate},dotted,thick] ({cos(42-0.66*84)},{sin(42-0.66*84)}) -- ({cos(90+0.67*180)},{sin(90+0.67*180)});
%
%
%
%
	
	%
	\filldraw [black] (-0.8,-0.7) circle (0pt) node[anchor=north] {\footnotesize $\KC$}; 	
	
	\filldraw [black] (2.5,0) circle (0.7pt);
	\draw [black] (2.5,0) circle (0.7pt) node[anchor=west] {$\frac{\Q_1}{v}$};

	\end{tikzpicture}
	\caption {\textbf{Left:} Bianchi type II solutions are heteroclinic orbits in the hemisphere $\BIIa$ with $\alpha$-limit sets within int$(A_1)$ and $\omega$-limit sets in $A_1^c$. \textbf{Right:} Projection of the Bianchi type II solutions into the $\Sigma$-plane.}\label{KC_3d_map}
\end{figure}

\subsection{Bianchi types $\mathrm{VI}_0$ and $\mathrm{VII}_0$}

To obtain the equations for the types $\mathrm{VI}_0$ and $\mathrm{VII}_0$ models we set, without loss of generality, 
$N_1=0$, $N_2>0$, $N_3<0$ for type $\mathrm{VI}_0$,
and $N_1=0$, $N_2>0$, $N_3>0$ for type $\mathrm{VII}_0$. Note that there are equivalent invariant subspaces when the sign is reversed.
We denote the type $\mathrm{VII}_0$ invariant subsets by $\BVIIa, \BVIIb, \BVIIc$, which means that each $N_\alpha N_\beta \neq 0$ and possess the same sign (containing both cases of positive and negative signs).
The Bianchi types $\mathrm{VI}_0$ and $\mathrm{VII}_0$ share several features and their dynamical analysis is a consequence of three monotone quantities, see \cite{HellLappicyUggla} for proofs.

Note that the constraint (in both cases) becomes
\begin{equation}\label{constrVIVII}
1 - \Sigma^2 - (N_2 - N_3)^2 =0, \qquad \text{ where } \qquad \Sigma^2 := \Sigma_+^2 + \Sigma_-^2.
\end{equation}
Due to the constraint~\eqref{constrVIVII}, the state spaces for the
types $\mathrm{VI}_0$ and $\mathrm{VII}_0$ models with $N_1=0$ are
3-dimensional with a 2-dimensional boundary given by the union of the
invariant type $\mathrm{II}_2$, $\mathrm{II}_3$ and $\mathrm{K}^\ocircle$
sets. Type $\mathrm{VI}_0$ has a relatively compact state space,
whereas type $\mathrm{VII}_0$ has an unbounded one.
Indeed, equation~\eqref{constrVIVII} implies that $\Sigma_+^2 + \Sigma_-^2\leq 1$.
For type $\mathrm{VI}_0$, $(N_2-N_3)^2= N_2^2 + N_3^2 + 2|N_2N_3|$, and
hence~\eqref{constrVIVII} yields $N_2^2 \leq 1 -\Sigma^2$ and
$N_3^2 \leq 1 -\Sigma^2$, where the equalities hold individually for the $\mathrm{II}_2$
and $\mathrm{II}_3$ boundary sets, respectively. For type $\mathrm{VII}_0$,
on the other hand, introducing $N_\pm := N_2\pm N_3$ results in that the
constraint~\eqref{constrVIVII} can be written as $\Sigma^2 + N_-^2 = 1$,
and thus $\Sigma_\pm$ and $N_-$ are bounded, while $N_+$ is unbounded.

More similar features of both Bianchi types are, for example, $\Sigma_+=-1/(2v)$ is a 2-dimensional invariant subset in the
supercritical case, $v\in(1/2,1)$, both for types $\mathrm{VI}_0$
and $\mathrm{VII}_0$. They also have some
common asymptotic features. In particular, they have the same
$\mathrm{II}_2\cup\mathrm{II}_3\cup\mathrm{K}^\ocircle$ boundary.
In the supercritical case the stable set in the
Kasner circle set $\mathrm{K}^\ocircle$ is
given by $S_{\mathrm{VI}_0} = S_{\mathrm{VII}_0} =
S_{\mathrm{VI}_0,\mathrm{VII}_0} := \mathrm{K}^\ocircle\backslash
\mathrm{int}(A_2 \cup A_3)$, which, due to that $N_1=0$, is different than the
set $S$ in the supercritical Bianchi types $\mathrm{VIII}$ and $\mathrm{IX}$ models in~\eqref{setS},
see Figure~\ref{FIG:SVIVII}, although both
$S_{\mathrm{VI}_0,\mathrm{VII}_0}$ and $S$ are defined as the sets where
type II heteroclinic chains end. In the subcritical case, types
$\mathrm{VI}_0$ and $\mathrm{VII}_0$ also share the region $A_2\cap A_3$ in
$\mathrm{K}^\ocircle$, where both $N_2$ and $N_3$ are unstable in
$\mathrm{int}(A_2\cap A_3)$.
In the critical case, $A_2\cap A_3$ reduces to the Taub point $\mathrm{T}_1$.
These features are illustrated in Figure~\ref{FIG:SVIVII}.
\begin{figure}[H]
\minipage[t]{0.39\textwidth}\centering
        \boxed{\footnotesize{\text{$v < \tfrac12$}}}\\[0.25cm]
\begin{subfigure}\centering
    \begin{tikzpicture}[scale=1.15,yscale=-1,decoration={markings,mark=at position 0.6 with {\arrow{latex}}},>=latex]



    \draw[color=gray, dotted] (2.46,1.42) -- (-0.55,0.84);
    \draw[color=white, ultra thick] (0.15,0.97) -- (-0.55,0.84);
    \draw[dotted, thick, postaction={decorate}] (0.15,0.97) -- (-0.55,0.84);

    \draw[color=gray, dotted] (-2.46,1.42) -- (0.91,0.41);
    \draw[color=white, ultra thick] (-0.55,0.84) -- (0.91,0.41);
    \draw[dotted, thick, postaction={decorate}] (-0.55,0.84) -- (0.91,0.41);

    \draw[color=gray, dotted] (2.46,1.42) -- (-0.75,-0.69);
    \draw[color=white, ultra thick] (0.91,0.41) -- (-0.75,-0.69);
    \draw[dotted, thick, postaction={decorate}] (0.91,0.41) -- (-0.75,-0.69);

    \draw [domain=0:6.28,variable=\t,smooth] plot ({sin(\t r)},{cos(\t r)});

    \draw [ultra thick, dotted, white, domain=2.3:4,variable=\t,smooth] plot ({sin(\t r)},{cos(\t r)});

    \draw [very thick, domain=-0.29:0.29,variable=\t,smooth] plot ({0.985*sin(\t r)},{0.985*cos(\t r)});

    \draw[color=gray,rotate=120,dashed] (0,-2.85) -- (0.98,-0.29);
    \draw[color=gray,rotate=120,dashed] (0,-2.85) -- (-0.98,-0.29);

    \draw[color=gray,rotate=240,dashed] (0,-2.85) -- (0.98,-0.29);
    \draw[color=gray,rotate=240,dashed] (0,-2.85) -- (-0.98,-0.29);

    \draw[rotate=240] (0,-2.85) circle (0.1pt) node[anchor=north] {\scriptsize{$\mathrm{Q}_2/v$}};
    \draw[rotate=120]  (0,-2.85) circle (0.1pt) node[anchor=north] {\scriptsize{$\mathrm{Q}_3/v$}};

    \node at (0,1.2) {\scriptsize{$A_2\cap A_3$}};

    \node at (0,-1.2) {\scriptsize{$S_{\mathrm{VI}_0,\mathrm{VII}_0}$}};

\end{tikzpicture}
\end{subfigure}
\endminipage\hfill
\minipage[t]{0.3\textwidth}\centering
        \boxed{\footnotesize{\text{$v = \tfrac12$}}}\\[0.25cm]
\begin{subfigure}\centering
    \begin{tikzpicture}[scale=1.15,yscale=-1,decoration={markings,mark=at position 0.6 with {\arrow{latex}}},>=latex]

    \draw[color=gray,->](0,0) -- (0,-0.53)  node[anchor= south] {\scriptsize{$\Sigma_{+}$}};
    \draw[color=gray,->] (0,0) -- (0.53,0)  node[anchor= west] {\scriptsize{$\Sigma_{-}$}};
    
    \draw[color=gray, dotted] (1.75,1) -- (-0.65,0.76);
    \draw[color=white, ultra thick] (0.5,0.87) -- (-0.65,0.76);
    \draw[dotted, thick, postaction={decorate}] (0.5,0.87) -- (-0.65,0.76);

    \draw[color=gray, dotted] (-1.75,1) -- (0.91,0.41);
    \draw[color=white, ultra thick] (-0.65,0.76) -- (0.91,0.41);
    \draw[dotted, thick, postaction={decorate}] (-0.65,0.76) -- (0.91,0.41);

    \draw[color=gray, dotted] (1.75,1) -- (-0.67,-0.75);
    \draw[color=white, ultra thick] (0.91,0.405) -- (-0.67,-0.75);
    \draw[dotted, thick, postaction={decorate}] (0.91,0.405) -- (-0.67,-0.75);

    \draw [domain=0:6.28,variable=\t,smooth] plot ({sin(\t r)},{cos(\t r)});

    \draw [ultra thick, dotted, white, domain=2.1:4.2,variable=\t,smooth] plot ({sin(\t r)},{cos(\t r)});

    \draw[color=gray,dashed,-] (-1.75,1) -- (1.75,1);

    \draw[color=gray,dashed,-] (-1.75,1) -- (-0.88,-0.49);
    \draw[color=gray,dashed,-] (1.75,1) -- (0.88,-0.49);

    \draw[rotate=240] (0,-2) circle (0.1pt) node[anchor=north] {\scriptsize{$2\mathrm{Q}_2$}};
    \draw[rotate=120]  (0,-2) circle (0.1pt) node[anchor=north] {\scriptsize{$2\mathrm{Q}_3$}};

    \filldraw [black] (0,1) circle (1.25pt) node[anchor= north] {\scriptsize{$\mathrm{T}_1$}};
    \filldraw [black] (0.88,-0.49) circle (1.25pt) node[anchor= south west] {\scriptsize{$\mathrm{T}_2$}};
    \filldraw [black] (-0.88,-0.49) circle (1.25pt)node[anchor= south east] {\scriptsize{$\mathrm{T}_3$}};


    \node at (0,-1.2) {\scriptsize{$S_{\mathrm{VI}_0,\mathrm{VII}_0}$}};
\end{tikzpicture}
\end{subfigure}
\endminipage\hfill
\minipage[t]{0.3\textwidth}\centering
        \boxed{\footnotesize{\text{$v > \tfrac12$}}}\\[0.25cm]
\begin{subfigure}\centering
    \begin{tikzpicture}[scale=1.15,yscale=-1,decoration={markings,mark=at position 0.7 with {\arrow{latex}}},>=latex]

    \draw[color=gray,dotted] (-1.17,0.67) -- (1.17,0.67);

    \draw[white, ultra thick] (-0.74,0.67) -- (0.74,0.67);
    \draw[dotted, thick, postaction={decoration={markings,mark=at position 0.41 with {\arrow[thick,color=black]{latex reversed}}},decorate}] (-0.75,0.67) -- (0.75,0.67);

    \draw[rotate=-120,color=gray,dotted] (0,-1.35) -- (-1,0);
    \draw[rotate=-120,color=gray,dotted] (-1.17,0.67) -- (-0.86,-0.5);

    \draw[rotate=-120,white,ultra thick] (-0.29,-0.95) -- (-1,0);
    \draw[rotate=-120,dotted, thick, postaction={decorate}] (-0.29,-0.95) -- (-1,0);

    \draw[rotate=-120,white,ultra thick] (-1,0) -- (-0.86,-0.5);
    \draw[rotate=-120,dotted, thick, postaction={decorate}] (-1,0) -- (-0.86,-0.5);

    \draw[color=gray, dotted] (-1.17,0.67) -- (0.91,0.41);
    \draw[color=gray, dotted] (1.17,0.67) -- (-0.39,-0.92);

    \draw[color=white, ultra thick] (-0.78,0.62) -- (0.91,0.41);
    \draw[color=white, ultra thick] (0.91,0.405) -- (-0.39,-0.92);

    \draw[dotted, thick, postaction={decorate}] (-0.78,0.62) -- (0.91,0.41);
    \draw[dotted, thick, postaction={decorate}] (0.91,0.405) -- (-0.39,-0.92);

    \draw [line width=0.1pt,domain=0:6.28,variable=\t,smooth] plot ({sin(\t r)},{cos(\t r)});

    \draw [ultra thick, dotted, white, domain=-0.26:0.26,variable=\t,smooth] plot ({sin(\t r)},{cos(\t r)});
    \draw [ultra thick, dotted, white, domain=1.8:4.5,variable=\t,smooth] plot ({sin(\t r)},{cos(\t r)});


    \draw[color=gray,rotate=120,dashed,-] (0,-1.35) -- (0.685,-0.75);
    \draw[color=gray,rotate=120,dashed,-] (0,-1.35) -- (-0.685,-0.75);

    \draw[color=gray,rotate=-120,dashed,-] (0,-1.35) -- (0.685,-0.75);
    \draw[color=gray,rotate=-120,dashed,-] (0,-1.35) -- (-0.685,-0.75);

    \filldraw [rotate=-120-12 ] (0,-1) circle (0.6pt); 
    \filldraw [rotate=-120-108 ] (0,-1) circle (0.6pt); 


    \node at (0,1.2) {\scriptsize{$S_{\mathrm{VI}_0,\mathrm{VII}_0}$}};
    \node at (0,-1.2) {\scriptsize{$S_{\mathrm{VI}_0,\mathrm{VII}_0}$}};


    \draw[rotate=240] (0,-1.35) circle (0.1pt) node[anchor=east] {\scriptsize{$\mathrm{Q}_2/v$}};
    \draw[rotate=120]  (0,-1.35) circle (0.1pt) node[anchor=west] {\scriptsize{$\mathrm{Q}_3/v$}};

\end{tikzpicture}
\end{subfigure}
\endminipage
\vspace{-0.6cm}
\captionof{figure}{The common stable set $S_{\mathrm{VI}_0,\mathrm{VII}_0}$
for Bianchi types $\mathrm{VI}_0$ and $\mathrm{VII}_0$. In addition, projected
onto $(\Sigma_+,\Sigma_-)$-space, there are illustrative heteroclinic chains
located on the $\mathrm{II}_2\cup\mathrm{II}_3\cup\mathrm{K}^\ocircle$ boundary.
In particular, $v\in(1/2,1)$ admits a heteroclinic cycle/chain with period 2,
which resides on the projected line between $\mathrm{Q}_2/v$ and $\mathrm{Q}_3/v$
characterized by $\Sigma_+=-1/(2v)$.
}\label{FIG:SVIVII}
\end{figure}
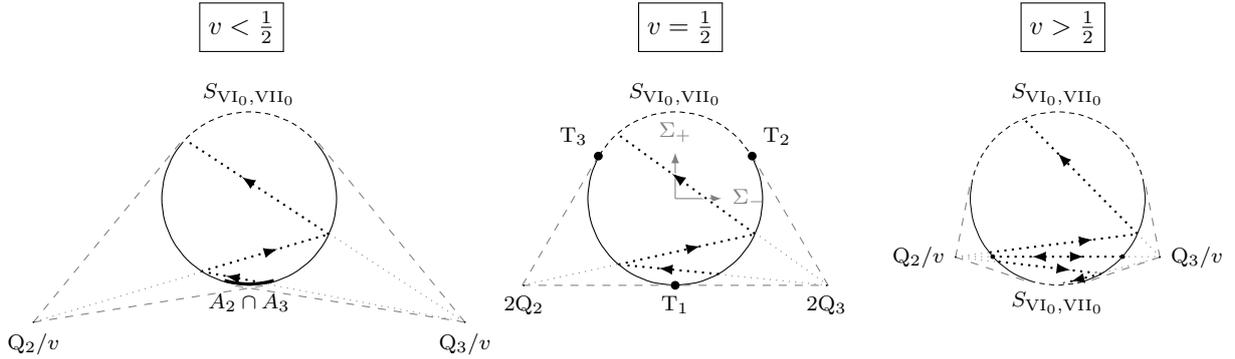
\begin{prop}\label{lem:genVIalpha}
In Bianchi type $\mathrm{VI}_0$ the limit sets (in $\tau_-$) are as follows:
\begin{itemize}
\item[(i)] When $v\in (0,1/2]$, the $\alpha$-limit set for all
orbits resides in the set $A_2\cap A_3\subseteq \mathrm{K}^\ocircle$,
where $A_2\cap A_3$ reduces to the Taub point $\mathrm{T}_1$ when $v=1/2$.
The $\omega$-limit set for all orbits resides in the set $S_{\mathrm{VI}_0}\subseteq \mathrm{K}^\ocircle$.
\item[(ii)] When $v\in (1/2,1)$, the $\alpha$-limit set for all orbits is
the fixed point $p_{\mathrm{VI}_0}$ given by
\begin{equation}\label{p_VI}
p_{\mathrm{VI}_0}:= \left\{ (\Sigma_+,\Sigma_-,N_2,N_3)=\left(-\frac{1}{2v},0, \frac{\sqrt{4v^2 - 1}}{4v},-\frac{\sqrt{4v^2 - 1}}{4v}\right) \right\}.
\end{equation}
%
The $\omega$-limit set of all orbits with $\Sigma_+\neq -1/(2v)$ resides in the set $S_{\mathrm{VI}_0}\subseteq \mathrm{K}^\ocircle$, whereas the $\omega$-limit set of all orbits on the invariant subset
$\Sigma_+ = - 1/(2v)$, apart from $p_{\mathrm{VI}_0}$, consists of the heteroclinic chain with period 2.
\end{itemize}
\end{prop}
In Proposition~\ref{lem:genVIalpha}~(ii), the invariant set
$\Sigma_+ = - 1/(2v)$ is conjugate to Bowen's eye:
the point $p_{\mathrm{VI}_0}$ is surrounded by spiraling orbits toward the heteroclinic chain with period 2, due to complex conjugate eigenvalues with positive real part, see~\cite{Tak94,Dutilleul}.
%
There are other special solutions of type $\mathrm{VI}_0$, which form the strong-unstable manifold of $p_{\mathrm{VI}_0}$. It consists of two heteroclinics from $p_{\mathrm{VI}_0}$ to the points $\mathrm{T}_1$ and $\mathrm{Q}_1$, which has the conserved quantities $\Sigma_-=0,N_2=-N_3$ and a monotone quantity $\Sigma_+<- 1/(2v)$ and  $\Sigma_+>- 1/(2v)$, respectively.

Let us now turn to type $\mathrm{VII}_0$, but before presenting asymptotic results
we first consider the locally rotationally symmetric (LRS) type
$\mathrm{VII}_0$ subset. This invariant set
is given by $N_-=0$ and $\Sigma_-=0$, where the constraint~\eqref{constrVIVII}
divides the LRS subset into two disjoint invariant sets consisting of the two
lines at $\Sigma_+ = 1$ and $\Sigma_+ = -1$, i.e.,
\begin{equation}\label{LRS}
\mathrm{LRS}^\pm := 
\left\{ (\Sigma_+,0,N_2,N_3) \in \mathbb{R}^4 \Bigm|
\begin{array}{c}
\,\, \Sigma_+=\pm 1, \\
\,\, N_2=N_3\neq 0
\end{array}
\right\},
\end{equation}
where the superscript of $\mathrm{LRS}^\pm$ is determined by the sign of $\Sigma_+$.
Let $N := N_2 = N_3>0$. Then the flow on the $\mathrm{LRS}^\pm$ subsets is determined as follows:
\begin{equation}\label{Neqn}
    N^\prime = -4\left(v \pm \frac{1}{2}\right)N.
\end{equation}
On $\mathrm{LRS}^+$, the variable $N\in (0,\infty)$ monotonically
decreases from  $\lim_{\tau_-\rightarrow -\infty}N =\infty$  to $0$, and hence
this orbit ends at $\mathrm{Q}_1\in \mathrm{K}^\ocircle$
for all $v\in(0,1)$. 
On $\mathrm{LRS}^-$, there are three $v$-dependent
cases: the critical case, $v=1/2$, which results in a line of fixed points; the
subcritical case, $v \in (0,1/2)$, which yields an orbit that emanates from
$\mathrm{T}_1$, where $N\in (0,\infty)$ subsequently monotonically increases,
resulting in $\lim_{\tau_-\rightarrow\infty}N = \infty$; the supercritical case,
$v \in (1/2,1)$, reverses the flow and leads to an orbit for which
$\lim_{\tau_-\rightarrow -\infty}N =\infty$, while it ends at
$\mathrm{T}_1$.
%
%
%
%
%
\begin{prop}\label{lem:genVIVIIalpha}
In Bianchi type $\mathrm{VII}_0$ the limit sets (in $\tau_-$) are as follows:
\begin{itemize}
\item[(i)] When $v \in (0,1/2)$, the $\alpha$-limit set of all non-$\mathrm{LRS}^+$ orbits reside in the set $A_2\cap A_3$, whereas $\mathrm{LRS}^+$ consists of an
orbit such that $\lim_{\tau_-\rightarrow-\infty} N = \infty$, where $N:=N_2=N_3$.
The $\omega$-limit set for all orbits resides in the stable set $S_{\mathrm{VII}_0}\subseteq \mathrm{K}^\ocircle$, apart from the $\mathrm{LRS}^-$ set, which consists of an orbit for which $\lim_{\tau_-\rightarrow\infty} N = \infty$.
\item[(ii)] When $v=1/2$, the $\alpha$-limit set of all non-$\mathrm{LRS}^+$ orbits
is the line of fixed points $\mathrm{LRS}^-$, $N_2=N_3=\mathrm{constant}$,
whereas $\mathrm{LRS}^+$ consists of an
orbit for which $\lim_{\tau_-\rightarrow-\infty} N = \infty$.
The $\omega$-limit set for all orbits, apart from the $\mathrm{LRS}^-$ set, lies in the stable set $S_{\mathrm{VII}_0}\subseteq \mathrm{K}^\ocircle$.
\item[(iii)] When $v\in (1/2,1)$, all non-$\mathrm{LRS}$ orbits asymptotically
satisfy
\begin{equation}
\lim_{\tau_-\rightarrow-\infty}\Sigma_+=-\frac{1}{2v}, \qquad \lim_{\tau_-\rightarrow-\infty}N_+=\infty, \qquad N_+:=N_2+N_3,
\end{equation}
whereas $(\Sigma_-,N_-)=(\sqrt{1-\Sigma_+^2}\cos\psi,\sqrt{1-\Sigma_+^2}\sin\psi)$ are asymptotically oscillatory, since $\psi$ is strictly monotone as $\tau_-\rightarrow -\infty$. Each $\mathrm{LRS}^\pm$ set consists of an orbit
such that $\lim_{\tau_-\rightarrow-\infty} N = \infty$.
The $\omega$-limit set for all orbits resides in the stable set $S_{\mathrm{VII}_0}\subseteq \mathrm{K}^\ocircle$, apart from the invariant set of co-dimension
one, characterized by $\Sigma_+=-1/(2v)$, for which the heteroclinic cycle of
period 2 on the $\mathrm{II}_2\cup\mathrm{II}_3\cup\mathrm{K}^\ocircle$
boundary is the $\omega$-limit.
\end{itemize}
\end{prop}
Since the type $\mathrm{VII}_0$ subset is unbounded, one could use a Poincaré compactification to describe its dynamical structure at infinity, as it could play a role in the generic dynamics; 
see e.g. \cite{Gao} for such usage in cosmology. 
Indeed, the oscillatory behavior at infinity in case (iii) above is described by a periodic orbit or a heteroclinic cycle, and its change of stability at $v=1/2$ would be critical when proving an attractor theorem for $v\in (0,1/2)$.
%
%
%
%
\subsection{Bianchi types $\mathrm{VIII}$ and $\mathrm{IX}$}\label{sec:proof}

Another monotone function for the types $\mathrm{VIII}$ and $\mathrm{IX}$ models is given by
\begin{equation}\label{Delta}
\Delta := 3|N_1N_2N_3|^{2/3} \qquad \text{ that satisfies } \qquad \Delta^\prime = -8v\Sigma^2\Delta.
\end{equation}
%
%
%
Thus $\Delta$ is monotonically decreasing when $\Sigma^2 >0$,
and has an inflection point when $\Sigma^2 = 0$, since $\left.\Delta''\right|_{\Sigma^2 =0} = 0$ and $\left.\Delta'''\right|_{\Sigma^2 =0} = - \frac{8}{3}v\left({\cal S}_1^2 + {\cal S}_2^2 + {\cal S}_3^2\right)\Delta$, 
%
%
where ${\cal S}_1^2 + {\cal S}_2^2 + {\cal S}_3^2>0$.
Hence,
\begin{equation}\label{Delta0}
\lim_{\tau_-\rightarrow\infty}\Delta = 0.
\end{equation}
%
%
%
%
Due to~\eqref{Delta0}, at least one of the variables $N_\alpha$ decays to $0$ for all $v \in (0,1)$. Thus the $\omega$-limit set of all type $\mathrm{IX}$ orbits resides in the union of the closure of the type $\mathrm{VII}_0$ subsets,
whereas the $\omega$-limit set of all type $\mathrm{VIII}$ orbits lies in closure of the union
of one type $\mathrm{VII}_0$ subset and the two type $\mathrm{VI}_0$ subsets.
Next, we prove global existence and boundedness of solutions.
\begin{lem}\label{lem:bdd}
    For any $v\in (0,1)$, all type $\mathrm{VIII}$ and $\mathrm{IX}$ solutions of~\eqref{intro_dynsyslambdaR} exist for all $\tau_-\in \mathbb{R}_+$ and satisfy  $ -\Delta \leq \Omega_k \le 1$ and $\Sigma^2 \leq 1 + \Delta$. In particular,  $0 \leq \lim_{\tau_-\rightarrow\infty}\Omega_k \leq 1$ and $\lim_{\tau_-\rightarrow\infty}\Sigma^2 \leq 1$. Moreover, if $v\in (1/2,1)$, then $N_\alpha,\alpha=1,2,3$ are also bounded for $\tau_-\in \mathbb{R}_+$.
\end{lem}
\begin{pf}
%
Since $\Sigma^2\geq 0$, the constraint~\eqref{intro_cons1} implies that  $\Omega_k \leq 1$. Next, we will show that $\Omega_k$ in~\eqref{Omega_k} and $\Delta$ in~\eqref{Delta} satisfy
$\Omega_k + \Delta \geq 0$. Indeed, we apply Schur's inequality given by:
\begin{equation}
A^3 + B^3 + C^3 + 3ABC \geq
AB(A+B) + AC(A+C) + BC(B+C),
\end{equation}
with $A = N_1^{2/3}$, $B = N_2^{2/3}$, $C = N_3^{2/3}$, which amounts to:
\begin{align}\label{ineqX}
N_1^2 + N_2^2 + N_3^2 + 3 (N_1N_2N_3)^{2/3}
\geq \,
& (N_1N_2)^{2/3} (N_1^{2/3}+ N_2^{2/3})\nonumber\\
& + (N_1N_3)^{2/3} (N_1^{2/3}+ N_3^{2/3}) \\
& + (N_2N_3)^{2/3} (N_2^{2/3}+ N_3^{2/3}).\nonumber
\end{align}
We apply the arithmetic mean and geometric mean (AM-GM) inequality to each line of the right-hand side of \eqref{ineqX}, e.g. $N_i^{2/3} + N_j^{2/3} \geq 2(N_i^{2/3}N_j^{2/3})^{1/2}=2 (N_iN_j)^{1/3}$, which yields:
\begin{equation}
N_1^2 + N_2^2 + N_3^2 + 3 (N_1N_2N_3)^{2/3}
\geq 2 (N_1N_2 + N_1N_3 + N_2N_3),
\end{equation}
and proves $\Omega_k \geq -\Delta$ as desired.
%
%
Thus, the constraint \eqref{intro_cons1} implies that
\begin{equation}\label{SigmaBOUND}
    \Sigma^2(\tau_{-}) 
    \leq 1 + \Delta (\tau_-) \leq 1 +\Delta (0) ,
\end{equation}
for all times $\tau_- \geq 0$, where we have used that $\Delta$ is non-increasing, due to \eqref{Delta}. In particular, $|\Sigma_\alpha (\tau_-)| \leq \sqrt{6(1+\Delta(0))}$ for each $\alpha=1,2,3$ and all $\tau_-\geq 0$.

To prove that $ N_\alpha,\alpha=1,2,3,$ are global in time, note \eqref{intro_dynsyslambdaR_N} yields $N_\alpha^\prime \leq 2 \sqrt{6 (1+\Delta(0))} N_\alpha$
and thus $N_\alpha(\tau_-)\leq \exp (\sqrt{6 (1+\Delta(0))} \tau_-) N_\alpha(0)$, implying that solutions exist for all $\tau_-\geq 0$; they could only blow up in infinite time.

Lastly, for any $v\in (1/2,1)$, we prove that $N_\alpha,\alpha=1,2,3$ are bounded for all $\tau_-\in \mathbb{R}_+$. Note that the type $\mathrm{VII}_0$ flow is dissipative, as solutions converge to a compact set, see Proposition~\ref{lem:genVIVIIalpha}. Since type $\mathrm{IX}$ solutions converge to the closure of the type $\mathrm{VII}_0$, due to~\eqref{Delta0}, the type $\mathrm{IX}$ solutions for sufficiently large times are nearby the type $\mathrm{VII}_0$ boundary and thereby can be seen as a perturbation of a dissipative flow - which is still dissipative. 
%
%
%
\qed
\end{pf}
This proof of boundedness of $N_\alpha,\alpha=1,2,3$ holds for $v\in (1/2,1)$, but fails for $v\in (0,1/2]$. Indeed, for $v=1/2$, some solutions converge to the unbounded line of equilibria $\mathrm{LRS}^-$; whereas this line is a blow-up solution for $v\in (0,1/2)$. Either way, the types $\mathrm{VI}_0$ and $\mathrm{VII}_0$ flows are not dissipative and therefore we cannot use such a perturbation argument.




\subsubsection{Non-generic Dynamics}

We now describe the dynamics of non-generic sets, which consists of the fixed point $p_{\mathrm{VI}_0}$ and the two-dimensional locally rotationally symmetric set $\mathcal{LRS}_\alpha$.

\textbf{The fixed point $p_{\mathrm{VI}_0}$}

On one hand, for $v\in (0, 1/2)$, note that $p_{\mathrm{VI}_0}$ does not play a role in the type $\mathrm{VI}_0$ dynamics, where $|\Sigma_+|<1$, which is not the case for this point. 
On the other hand, for $v\in (1/2,1)$, it lies within the constraints and thus we linearize the system \eqref{full:subs} at it. 
This yields one negative eigenvalue, $-2/v$, corresponding to the stable invariant set defined by the constraint, as its associated eigenvector is normal to the constraint. The remaining 4 eigenvalues are:
\begin{equation}
    \frac{4v^2-1 \pm  \sqrt{16v^4 - 56v^2 + 13}}{2v}, \qquad  \frac{4v^2-1}{v},\qquad -\frac{3}{v},
\end{equation}
where the three eigenvalues with positive real part correspond to its 3d unstable manifold consisting of the type $\mathrm{VI}_0$ subset (where the pair of complex eigenvalues correspond to the 2d weak-unstable submanifold given by $\Sigma_+=-1/(2v)$ and the other eigenvalue yields the 1d strong-unstable submanifold given by $\Sigma_-=0,N_2=-N_3$); and 
the negative eigenvalue corresponds to a 1d strong-stable submanifold, which we discuss this next.
%
\begin{lem}\label{lem:strong-stable-pVI}
When $v\in (1/2,1)$, the  equilibrium \eqref{p_VI} (with additionally $N_1=0$) possesses a one-dimensional of strong-stable manifold of the flow \eqref{full:subs}-\eqref{constraint}, denoted by $W^{ss}(p_{\mathrm{VI}_0})$, which is tangent to the eigenvector (corresponding to the strong-stable eigenvalue $\lambda_{ss} = -3/v$):
\begin{equation}\label{1dSSEV}
    e_{ss} = (\Sigma_+, \Sigma_-, N_1, N_2, N_3) = \left( 0 ,  1, \dfrac{8v^2+1}{\sqrt{3(4v^2-1)}} ,  \dfrac{\sqrt{3(4v^2-1)}}{6} ,  \dfrac{\sqrt{3(4v^2-1)}}{6} \right).
\end{equation}
Locally, this manifold can be parametrized by $\Sigma_-$ according to:
%
%
%
\begin{subequations}\label{Wsspviv3}
\begin{alignat}{8}
\Sigma_+ &={} & -\frac{1}{2v} \quad 
           &+ \qquad \quad 0 &  
           &+ \, O(\Sigma_-^2), \\
N_1 &={} & 0 \quad\,\, 
           &+ \frac{8v^2+1}{\sqrt{3(4v^2-1)}} &\Sigma_- 
           &+ \, O(\Sigma_-^2), \\
N_2 &={} & \frac{\sqrt{4v^2-1}}{4v} 
      &+ \frac{\sqrt{3(4v^2-1)}}{6} &\Sigma_- 
      &+ \, O(\Sigma_-^2), \\
N_3 &={} & -\frac{\sqrt{4v^2-1}}{4v} 
      &+ \frac{\sqrt{3(4v^2-1)}}{6} &\Sigma_- 
      &+ \, O(\Sigma_-^2).
\end{alignat}
\end{subequations}
%
%
\end{lem}
\begin{pf}
    We obtain the existence of the strong-stable manifold by standard methods, due to hyperbolicity and a spectral-gap.
    The first two columns in equation~\eqref{Wsspviv3} are the usual compatibility conditions for the graph of the stable manifold: at the fixed point (when $\Sigma_-=0$), the graph should yield the fixed point \eqref{p_VI}; and the graph should be tangent to the strong-stable eigenvector in~\eqref{1dSSEV}. One could also compute the second order expansions by matching coefficients in the invariance condition, but we refrain from doing so here.
    \qed
\end{pf}
%
%

\textbf{Locally rotationally symmetric set $\mathcal{LRS}_\alpha$}

The dynamical system~\eqref{intro_dynsyslambdaR} for Bianchi types $\mathrm{VIII}$ and $\mathrm{IX}$ 
admits (physically equivalent) invariant locally rotationally symmetric (LRS) sets. For type $\mathrm{IX}$, they are given by
\begin{equation}\label{LRSIXdef}
\begin{aligned}
\mathcal{LRS}_\alpha &:=
\left\{(\Sigma_\alpha,\Sigma_\beta,\Sigma_\gamma,N_\alpha,N_\beta,N_\gamma) \in \mathbb{R}^6 \Bigm|
\begin{array}{c}
\,\, \Sigma_\beta=\Sigma_\gamma, \, N_\alpha>0, \, N_\beta=N_\gamma>0 \\
\,\, \text{ satisfying \eqref{intro_cons1}-\eqref{intro_cons2}}
\end{array}
\right\},
\end{aligned}
\end{equation}
where $(\alpha\beta\gamma)$ is a permutation of $(123)$, while type $\mathrm{VIII}$ only
admits a single LRS set since one of the variables $N_1$, $N_2$, $N_3$ has an opposite sign compared
to the other two. Without loss of generality, we use Misner variables according to $\Sigma_\alpha=-2\Sigma_+$ and $\Sigma_-=\Sigma_\beta-\Sigma_\gamma=0$:
\begin{equation}\label{LRS1IXdef}
\begin{aligned}
\mathcal{LRS}_1 &:=
\left\{(\Sigma_+,0,N_1,N_2,N_3) \in \mathbb{R}^5 \Bigm|
\begin{array}{c}
\,\,  N_1>0, \, N_2=N_3>0 \\
\,\, 1 = \Sigma_{+}^2 + N_1^2 - 4N_1N_2
\end{array}
\right\}.
\end{aligned}
\end{equation}
Note that this consists of an open set with boundary given by the two lines $\mathrm{LRS}^\pm$ in \eqref{LRS} (whenever $N_1=0$ and thus $\Sigma_{+}=\pm 1$), and the arc (whenever $N_2=0$ and thus $1 = \Sigma_{+}^2 + N_1^2$) that consists of the heteroclinic from $\mathrm{Q}_1$ to $\mathrm{T}_1$.
The LRS types $\mathrm{VIII}$ and $\mathrm{IX}$ sets have three distinct 
dynamical regimes, the subcritical, critical and supercritical cases. 
The analysis of the set in~\eqref{LRSIXdef} was carried out for the critical case, when $v=1/2$, in \cite{LRS}; see also \cite[Section 11]{Ringstrom} and \cite[Proposition 1]{Ringstrom2}.
We give a different proof of these results. 

The set $\mathcal{LRS}_1$ has the following skew-product flow, since $(\Sigma_+,N_1)$ decouples from $N_2$.:
\begin{subequations}\label{LRSIXeq}
    \begin{align}
            \Sigma'_{+} &= 4v (1 - \Sigma_{+}^2) \left(\Sigma_{+} - \frac{1}{4v}\right) - 3N_1^2,\label{LRSIXeq1}\\
            N_{1}' &= -4v \left(\Sigma_{+} - \frac{1}{v}\right)\Sigma_{+}N_{1},\label{LRSIXeq2}\\
            N_{2}' &= -4v \left(\Sigma_{+} + \frac{1}{2v}\right) \Sigma_{+}N_{2},\label{LRSIXeq3} 
    \end{align}
\end{subequations}
with the constraint $ \Sigma_{+}^2 + N_1^2 - 4N_1N_2 = 1$. 
Since $N_1,N_2>0$, we have $\Sigma_{+}^2 + N_1^2 > 1$. 
When $(\Sigma_+,N_1)=(\pm 1,0)$, we obtain the flow \eqref{Neqn} in the boundary set $\mathrm{LRS}^\pm$ of $\mathcal{LRS}_1$, defined in \eqref{LRS}. 
Next, we describe the asymptotics of these equations; see Figure~\ref{fig:LRS3d}.
\begin{figure}[H]
    \centering
    \begin{minipage}{0.32\textwidth}
        \centering
        \boxed{\footnotesize{\text{$v < \tfrac12$}}}\\[0.25cm]
        \includegraphics[width=\textwidth]{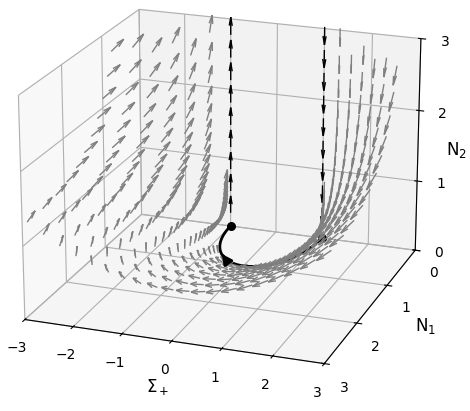}
    \end{minipage}
    %
    %
    \begin{minipage}{0.32\textwidth}
        \centering
        \boxed{\footnotesize{\text{$v =  \tfrac12$}}}\\[0.25cm]
        \includegraphics[width=\textwidth]{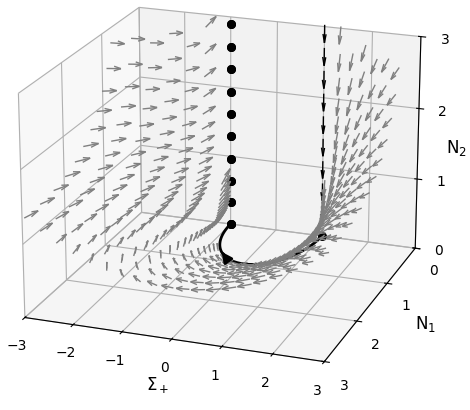}
    \end{minipage}
    %
    %
    \begin{minipage}{0.32\textwidth}
        \centering
        \boxed{\footnotesize{\text{$v >  \tfrac12$}}}\\[0.25cm]
        \includegraphics[width=\textwidth]{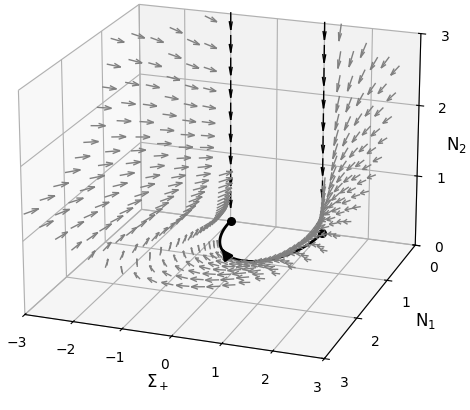}
    \end{minipage}
    \caption{The constrained two-dimensional dynamics of $\mathcal{LRS}_1$ (gray) described by the equation \eqref{LRSIXeq} and its corresponding boundary sets (bold): the lines $\mathrm{LRS}^\pm$ in \eqref{LRS} and the heteroclinic from $\mathrm{Q}_1$ to $\mathrm{T}_1$ within $\BIIa$ in \eqref{BII_N_1}. 
    The dynamics of each case (subcritical, critical or supercritical) is described in Lemma~\ref{lem:LRS3d}
    }
    \label{fig:LRS3d}
\end{figure}
\begin{lem} \label{lem:LRS3d}
    In $\mathcal{LRS}_1\subseteq \mathbb{R}^5$, the $\omega$-limit sets (in $\tau_-$) of all solutions are given as follows: 
    \begin{itemize}
        \item[(i)] When $v \in (0,1/2)$, all solutions blow-up according to $\lim_{\tau_-\to\infty} (\Sigma_+,N_1,N_2) =(-1,0,\infty)$.
        \item[(ii)] When $v =1/2$, all solutions converge to an equilibrium in the line $\mathrm{LRS}^-$ given by~\eqref{LRS}.
        \item[(iii)] When $v \in (1/2,1)$, all solutions converge to $\mathrm{T}_1$, given by $(\Sigma_+,N_1,N_2)=(-1,0,0)$. 
    \end{itemize}
    Moreover, in $\mathcal{LRS}_1\subseteq \mathbb{R}^5$, the $\alpha$-limit sets (in $\tau_-$) of all solutions are given as follows: 
    \begin{itemize}
        \item[(i)] When $v \in (0,1/4]$, there is one solution that blows-up (backwards in time) as in $\lim_{\tau_-\to-\infty} (\Sigma_+,N_1,N_2) =(-1/(4v), 0 , \infty)$, which is a separatrix: on one side, all solutions converge (backwards) to  $\mathrm{Q}_1$, given by $(\Sigma_+,N_1,N_2)=(1,0,0)$; whereas on the other side, all solutions blow-up (backwards) as in $\lim_{\tau_-\to-\infty} (\Sigma_+,N_1,N_2) =(\infty, \infty , \infty)$.
        \item[(ii)] When $v \in (1/4,1)$, all solutions satisfy $\lim_{\tau_-\to\infty} (\Sigma_+,N_1,N_2) =(\infty, \infty , \infty)$.
    \end{itemize}
%
        %
%
%
%
%
\end{lem}
\begin{pf}
For $N_1=0$, the constraint implies $\Sigma_+=\pm 1$ and thus we obtain the $\mathrm{LRS}^\pm$ subsets defined in \eqref{LRS}, which are the boundary of the subset $\mathcal{LRS}_1$. 
For $N_1>0$, we solve the constraint through 
\begin{equation}\label{N2}
    N_2=\frac{\Sigma_+^2-1}{4N_1}+\frac{N_1}{4}    
\end{equation}
and we obtain the decoupled two-dimensional flow given by \eqref{LRSIXeq1}-\eqref{LRSIXeq2}, which amounts to a skew-product flow when coupling with \eqref{LRSIXeq3}. 

We now describe the asymptotic dynamics of the planar system \eqref{LRSIXeq1}-\eqref{LRSIXeq2}, with vector field plotted in Figure~\ref{fig:Planar_flow}, and its implications for the constrained surface of the three-dimensional phase-space of \eqref{LRSIXeq}, which is obtained through \eqref{N2}.
Even though the constraint forces $\Sigma_{+}^2 + N_1^2 > 1$, we analyze the whole phase-space, as it indicates when qualitative changes occur as the parameter changes.

%
Note that there are three fixed points of \eqref{LRSIXeq1}-\eqref{LRSIXeq2}, 
which occur when $N_1 = 0$:
\begin{equation}
    \Sigma_+ = \pm 1 \qquad \text{or} \qquad \Sigma_+ = \frac{1}{4v}.
\end{equation}
%
%
We analyze the stability of equilibria by linearizing equations \eqref{LRSIXeq1}-\eqref{LRSIXeq2} and computing the associated eigenvalues. 
%
%
The eigenvalues of the equilibrium $(\Sigma_+, N_1) =(-1,0)$ are $\{-4(1+v), -2(1+4v)\}$, which implies this point is asymptotically stable for all $v\in (0,1)$. 
For the equilibrium $(\Sigma_+, N_1) = (1,0)$, the eigenvalues are $\{4(1-v), 8(1/4-v)\}$, and thus there is one positive eigenvalue for all $v\in (0,1)$ and another that is positive for $v \in (0,1/4)$, zero for $v=1/4$, and negative for $v \in (1/4,1)$. 
Similarly, the point $(\Sigma_+, N_1) =\left(1/(4v), 0\right)$ has eigenvalues $\left\{3/(4v), 4(v^2 - 1/16)/v\right\}$, and thus there is one positive eigenvalue for all $v\in (0,1)$ and another that is negative for $v \in (0,1/4)$, zero for $v=1/4$, and positive for $v \in (1/4,1)$. 
This occurs due to a transcritical bifurcation of the last two points at $v=1/4$, which indicates the intrusion of a fixed point that does not satisfy the constraint to the constrained phase-space. 
See Figure~\ref{fig:Planar_flow}. 

For both equilibria, $(\Sigma_+, N_1) =\left(1/(4v), 0\right)$ and $\left(1, 0\right)$, and all $v\in (0,1)$, 
the eigenvalue that changes sign has an eigendirection corresponding to the invariant set $N_1=0$ with flow $\Sigma_{+}' = 4v(1-\Sigma_{+}^2)(\Sigma_{+} - 1/(4v))$, and thus such transcritical bifurcation occurs within the invariant subspace $N_1=0$. 
The positive eigenvalue has an eigendirection with $N_1>0$, and thus these fixed points have an unstable manifold pointing in this direction.
See Figure~\ref{fig:Planar_flow}.
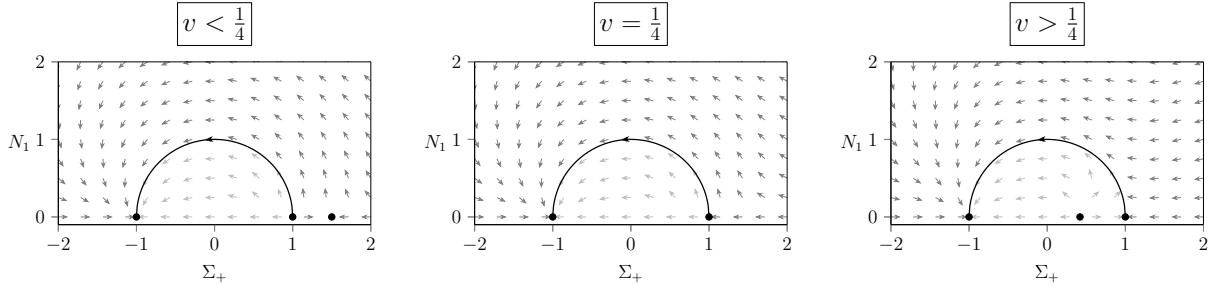
\begin{figure}[H]
    \centering
    \begin{tikzpicture}[scale=0.6]
        \begin{axis}[
            xmin = -2, xmax = 2,
            ymin = -0.1, ymax = 2,
            zmin = 0, zmax = 1,
            axis equal image,
            view = {0}{90},
            xlabel = $\Sigma_{+}$, 
            ylabel = $N_1$, 
            ylabel style = {
                rotate=270,
                at={(axis description cs:-0.125,0.5)}, 
                anchor=center}, 
            xtick = {-2.0, -1.0, 0,  1.0, 2.0},
            xtick pos=bottom,  
            ytick pos=left,    
            tick align=outside, 
            title = {\boxed{\Large{\text{$v <  \tfrac14$}}}},
            ]
            \addplot3[
                quiver={
                    u = {((1 - x^2)*(4*0.165*x - 1) - 3*y^2 ) /
                         sqrt( ((1 - x^2)*(4*0.165*x - 1) - 3*y^2)^2 +
                               (-4*(0.165*x - 1)*x*y)^2 + 1e-6 )},
                    v = { ( -4*(0.165*x - 1)*x*y ) /
                         sqrt( ((1 - x^2)*(4*0.165*x - 1) - 3*y^2)^2 +
                               (-4*(0.165*x - 1)*x*y)^2 + 1e-6 )},
                    scale arrows = 0.12
                },
                draw=gray,
                -{Stealth[length=3.5pt]},
                domain=-2:2,
                domain y=0:2,
                samples=15,
                samples y=9,
                restrict expr to domain={x^2 + y^2}{1:1000}, 
            ]
            {0};
            
            \addplot3[
                quiver={
                    u = {((1 - x^2)*(4*0.165*x - 1) - 3*y^2 ) /
                         sqrt( ((1 - x^2)*(4*0.165*x - 1) - 3*y^2)^2 +
                               (-4*(0.165*x - 1)*x*y)^2 + 1e-6 )},
                    v = { ( -4*(0.165*x - 1)*x*y ) /
                         sqrt( ((1 - x^2)*(4*0.165*x - 1) - 3*y^2)^2 +
                               (-4*(0.165*x - 1)*x*y)^2 + 1e-6 )},
                    scale arrows = 0.12
                },
                draw=lightgray,
                -{Stealth[length=3.5pt]},
                domain=-2:2,
                domain y=0:2,
                samples=15,
                samples y=9,
                restrict expr to domain={x^2 + y^2}{0:1}, 
            ]
            {0};

            \addplot+[
                only marks,
                mark=*,
                mark size=2pt,
                mark options={fill=black, draw=black}
            ] coordinates {
                (-1.00, 0)
                (1.00, 0)
                (1.5, 0) 
            };
            
        \node at (axis cs:1.45,0) [circle, draw=none, fill=white, minimum size=0.5pt] {};;

        \addplot[
            domain=1:-1,      
            samples=200,
            thick,
            black,
            postaction={
                decorate,
                decoration={
                    markings,
                    mark=at position 0.54 with {\arrow{Stealth[length=5pt]}}
                }
            }
        ]
        ({x}, {sqrt(1 - x^2)});
        
        \end{axis}
    \end{tikzpicture}
    \hspace{0.2cm}
    \begin{tikzpicture}[scale=0.6]
        \begin{axis}[
            xmin = -2, xmax = 2,
            ymin = -0.1, ymax = 2,
            zmin = 0, zmax = 1,
            axis equal image,
            view = {0}{90},
            xlabel = $\Sigma_{+}$, 
            ylabel = $N_1$, 
            ylabel style = {
                rotate=270,
                at={(axis description cs:-0.125,0.5)}, 
                anchor=center}, 
            xtick = {-2.0, -1.0, 0, 1.0, 2.0},
            xtick pos=bottom,  
            ytick pos=left,    
            tick align=outside, 
            title = {\boxed{\Large{\text{$v =  \tfrac14$}}}},
        ]
        
        \addplot3[
            quiver={
                u = {((1 - x^2)*(4*0.25*x - 1) - 3*y^2 ) /
                     sqrt( ((1 - x^2)*(4*0.25*x - 1) - 3*y^2)^2 +
                           (-4*(0.25*x - 1)*x*y)^2 + 1e-6 )},
                v = { ( -4*(0.25*x - 1)*x*y ) /
                     sqrt( ((1 - x^2)*(4*0.25*x - 1) - 3*y^2)^2 +
                           (-4*(0.25*x - 1)*x*y)^2 + 1e-6 )},
                scale arrows=0.12
            },
            draw=lightgray,
            -{Stealth[length=3.5pt]},
            domain=-2:2,
            domain y=0:2,
            samples=15,
            samples y=9,
            restrict expr to domain={x^2 + y^2}{0:1},
        ]{0};
        
        \addplot3[
            quiver={
                u = {((1 - x^2)*(4*0.25*x - 1) - 3*y^2 ) /
                     sqrt( ((1 - x^2)*(4*0.25*x - 1) - 3*y^2)^2 +
                           (-4*(0.25*x - 1)*x*y)^2 + 1e-6 )},
                v = { ( -4*(0.25*x - 1)*x*y ) /
                     sqrt( ((1 - x^2)*(4*0.25*x - 1) - 3*y^2)^2 +
                           (-4*(0.25*x - 1)*x*y)^2 + 1e-6 )},
                scale arrows=0.12
            },
            draw=gray,
            -{Stealth[length=3.5pt]},
            domain=-2:2,
            domain y=0:2,
            samples=15,
            samples y=9,
            restrict expr to domain={x^2 + y^2}{1:1000},
        ]{0};

        \addplot+[
            only marks,
            mark=*,
            mark size=2pt,
            mark options={fill=black, draw=black}
        ]   coordinates {
                (-1.00, 0)
                (1.00, 0)
            };


        \addplot[
            domain=1:-1,      
            samples=200,
            thick,
            black,
            postaction={
                decorate,
                decoration={
                    markings,
                    mark=at position 0.54 with {\arrow{Stealth[length=5pt]}}
                }
            }
        ]
        ({x}, {sqrt(1 - x^2)});
        
        \end{axis}
    \end{tikzpicture}
    \hspace{0.2cm}
    \begin{tikzpicture}[scale=0.6]
        \begin{axis}[
            xmin = -2, xmax = 2,
            ymin = -0.1, ymax = 2,
            zmin = 0, zmax = 1,
            axis equal image,
            view = {0}{90},
            xlabel = $\Sigma_{+}$, 
            ylabel = $N_1$, 
            ylabel style = {
                rotate=270,
                at={(axis description cs:-0.125,0.5)}, 
                anchor=center}, 
            xtick = {-2.0, -1.0, 0, 1.0, 2.0},
            xtick pos=bottom,  
            ytick pos=left,    
            tick align=outside, 
            title = {\boxed{\Large{\text{$v >  \tfrac14$}}}},
        ]
        
        \addplot3[
            quiver={
                u = {((1 - x^2)*(4*0.8*x - 1) - 3*y^2 ) /
                     sqrt( ((1 - x^2)*(4*0.8*x - 1) - 3*y^2)^2 +
                           (-4*(0.8*x - 1)*x*y)^2 + 1e-6 )},
                v = { ( -4*(0.8*x - 1)*x*y ) /
                     sqrt( ((1 - x^2)*(4*0.8*x - 1) - 3*y^2)^2 +
                           (-4*(0.8*x - 1)*x*y)^2 + 1e-6 )},
                scale arrows=0.12
            },
            draw=lightgray,
            -{Stealth[length=3.5pt]},
            domain=-2:2,
            domain y=0:2,
            samples=15,
            samples y=9,
            restrict expr to domain={x^2 + y^2}{0:1},
        ]{0};
        
        \addplot3[
            quiver={
                u = {((1 - x^2)*(4*0.8*x - 1) - 3*y^2 ) /
                     sqrt( ((1 - x^2)*(4*0.8*x - 1) - 3*y^2)^2 +
                           (-4*(0.8*x - 1)*x*y)^2 + 1e-6 )},
                v = { ( -4*(0.8*x - 1)*x*y ) /
                     sqrt( ((1 - x^2)*(4*0.8*x - 1) - 3*y^2)^2 +
                           (-4*(0.8*x - 1)*x*y)^2 + 1e-6 )},
                scale arrows=0.12
            },
            draw=gray,
            -{Stealth[length=3.5pt]},
            domain=-2:2,
            domain y=0:2,
            samples=15,
            samples y=9,
            restrict expr to domain={x^2 + y^2}{1:1000},
        ]{0};

        \addplot+[
            only marks,
            mark=*,
            mark size=2pt,
            mark options={fill=black, draw=black}
        ] coordinates {
            (-1.00, 0)
            (1.00, 0)
            (0.42, 0)
        };

        \addplot[
            domain=1:-1,      
            samples=200,
            thick,
            black,
            postaction={
                decorate,
                decoration={
                    markings,
                    mark=at position 0.54 with {\arrow{Stealth[length=5pt]}}
                }
            }
        ]
        ({x}, {sqrt(1 - x^2)});        
        \end{axis}
    \end{tikzpicture}
    
    \caption{The qualitative description of the projected two-dimensional dynamics of the $\mathcal{LRS}_1$ subsets given by the equations \eqref{LRSIXeq1}-\eqref{LRSIXeq2}.  
    Note that the constraint forces $\Sigma_{+}^2 + N_1^2 > 1$, where its boundary consists of $\Sigma_{+}^2 + N_1^2 = 1$, which amounts to the heteroclinic from $\mathrm{Q}_1$ to $\mathrm{T}_1$ within $\BIIa$ and the fixed points $(\Sigma_+, N_1) =\left(\pm 1, 0\right)$, which are the projected lines $\mathrm{LRS}^\pm$.
    There is a third fixed point, $(\Sigma_+, N_1) =\left(1/(4v), 0\right)$, which only lies within the constrained phase-space when $v<1/4$.
    }
    \label{fig:Planar_flow}
\end{figure}

Next, we show there are no periodic orbits of \eqref{LRSIXeq1}-\eqref{LRSIXeq2} for any $v\in (0,1)$. 
Indeed, we split the phase-space in regions, denoted by (a), (b), (c), (d), which are delimited by the isoclines, given by $\Sigma_+=-1,1/(4v),+1$, see Figure~\ref{fig:isoclines}. There are three cases:
\begin{subequations}
\begin{alignat}{2}
\boxed{\footnotesize{\text{$v < \tfrac14$}}}\quad 
&\text{(a)}\; \Sigma_+ < -1,\;
\text{(b)}\; \Sigma_+ \in [-1,1],\;
\,\,\,\,\,\text{(c)}\; \Sigma_+ \in \left(1,\tfrac{1}{4v}\right),\;
&\,\,\,\text{(d)}\; \Sigma_+ \ge \tfrac{1}{4v}, \\[6pt]
\boxed{\footnotesize{\text{$v = \tfrac14$}}}\quad 
&\text{(a)}\; \Sigma_+ < -1,\;
\text{(b)}\; \Sigma_+ \in [-1,1],\;
&\text{(d)}\; \Sigma_+ \ge 1, \\[6pt]
\boxed{\footnotesize{\text{$v > \tfrac14$}}}\quad 
&\text{(a)}\; \Sigma_+ < -1,\;
\text{(b)}\; \Sigma_+ \in \left[-1,\tfrac{1}{4v}\right],\;
\text{(c)}\; \Sigma_+ \in \left(\tfrac{1}{4v},1\right),\;
&\text{(d)}\; \Sigma_+ \ge 1.
\end{alignat}
\end{subequations}
%
%
%
%
There cannot be a periodic orbit in more than two regions of  (a), (b), (c), (d), as the flow is monotone, $\Sigma_+' <0$, at each of the isoclines $\Sigma_+=\pm 1,1/(4v)$ for all $v\in (0,1)$.
%
Next we discuss each of the regions, see Figure~\ref{fig:isoclines}. There is no periodic orbit in (b) or (d), since the vector field is monotone: $\Sigma_+' <0$. 
Similarly, there is no periodic orbit in (c), since the vector field is monotone\footnote{As a matter of fact, equation \eqref{LRSIXeq1} amounts to $\Sigma'_{+} <  4(1 - \Sigma_{+}^2) \left(v\Sigma_{+} - 1\right)<0$ within the constraint $\Sigma_{+}^2 + N_1^2 > 1$ in the region (c).}: $N_1' >0$.
For (a), let $f(\Sigma_+,N_1)$ denote the vector field of \eqref{LRSIXeq1}-\eqref{LRSIXeq2}, hence the Dulac-Bendixson criterion implies that there are no periodic orbits in this region: 
\begin{equation}\label{Dulac}
    \text{div}\left(\frac{1}{1 - \Sigma_+^2} f(\Sigma_+,N_1)\right ) = 4v - \frac{6N_1^2\Sigma_+}{(1 + \Sigma_+^2)^2} - \frac{4\Sigma_+(v\Sigma_+ -1)}{1 - \Sigma_+^2} > 0,
\end{equation}
for $N_1 >0$ and $\Sigma_+ < -1 < 1/v$ for $v\in (0,1)$; and thus all terms in \eqref{Dulac} are positive. 
%
%
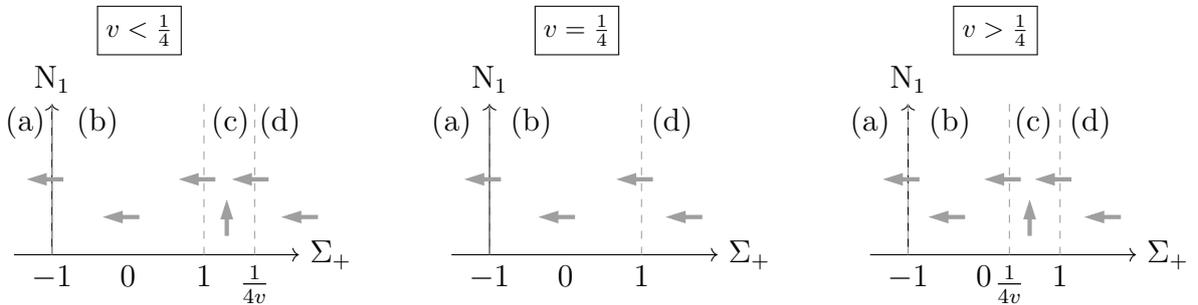
\begin{figure}[H]
    \centering
    \begin{tikzpicture}[scale = 0.5]
        \draw[->] (-3,0) -- (4.5,0) node[right] {$\Sigma_+$};
        \draw[->] (-2,0) -- (-2, 4) node[above] {N$_1$};
        \draw (-2,0) node[below]{$-1$} (0,0) node[below]{ $0$} (2,0) node[below]{$1$} (4-2/3,0) node[below]{$\frac{1}{4v}$};
        \draw (4,3.5) node {({d})};
        \draw (2.7,3.5) node {({c})};
        \draw (-0.8,3.5) node {({b})};
        \draw (-2.7,3.5) node {({a})};
        
        \draw[gray!75,dashed] (4-2/3,0) -- (4-2/3,4);
        \draw[gray!75,dashed] (2,0) -- (2,4);
        \draw[gray!75,dashed] (-2,0) -- (-2,4);
        

        \draw[arrows = {-Stealth[inset=0pt, length=8pt, angle'=30]},ultra thick,gray!75, ] (3.7,2)--(2.7,2);
        
        \draw[arrows = {-Stealth[inset=0pt, length=8pt, angle'=30]},ultra thick,gray!75, ] (5,1)--(4,1);

        \draw[arrows = {-Stealth[inset=0pt, length=8pt, angle'=30]},ultra thick,gray!75,](3-0.4,0.5) -- (3-0.4,1.5);
        
        \draw[arrows = {-Stealth[inset=0pt, length=8pt, angle'=30]},ultra thick,gray!75, ] (-1.7,2)--(-2.7,2);


        \draw[arrows = {-Stealth[inset=0pt, length=8pt, angle'=30]},ultra thick,gray!75, ] (2.3,2)--(1.3,2);

        \draw[arrows = {-Stealth[inset=0pt, length=8pt, angle'=30]},ultra thick,gray!75, ] (0.3,1)--(-0.7,1);

        \node[above,font=\footnotesize\bfseries] at (0.3,5) {\boxed{\text{$v <  \tfrac14$}}};
    \end{tikzpicture}    
    \hspace{0.5cm}
    \begin{tikzpicture}[scale = 0.5]
        \draw[->] (-3,0) -- (4,0) node[right] {$\Sigma_+$};
        \draw[->] (-2,0) -- (-2, 4) node[above] {N$_1$};
        \draw (-2,0) node[below]{$-1$} (0,0) node[below]{ $0$} (2,0) node[below]{$1$};
        \draw[white] (4-2/3,0) node[below]{$\frac{1}{4v}$};
        \draw (2.8,3.5) node {(d)};
        \draw (-0.9,3.5) node {(b)};
        \draw (-3,3.5) node {(a)};
        
        \draw[gray!75,dashed] (2,0) -- (2,4);
        \draw[gray!75,dashed] (-2,0) -- (-2,4);
        

        \draw[arrows = {-Stealth[inset=0pt, length=8pt, angle'=30]},ultra thick,gray!75, ] (2.3,2)--(1.3,2);
        
        \draw[arrows = {-Stealth[inset=0pt, length=8pt, angle'=30]},ultra thick,gray!75, ] (3.6,1)--(2.6,1);

       \draw[arrows = {-Stealth[inset=0pt, length=8pt, angle'=30]},ultra thick,gray!75, ] (0.25,1)--(-0.75,1);
        
        \draw[arrows = {-Stealth[inset=0pt, length=8pt, angle'=30]},ultra thick,gray!75, ] (-1.7,2)--(-2.7,2);

        \node[above,font=\footnotesize\bfseries] at (0.3,5) {\boxed{\text{$v =  \tfrac14 $}}};
    \end{tikzpicture}
    \hspace{0.5cm}
    \begin{tikzpicture}[scale = 0.5]
        \draw[->] (-3,0) -- (4,0) node[right] {$\Sigma_+$};
        \draw[->] (-2,0) -- (-2, 4) node[above] {N$_1$};
        \draw (-2,0) node[below]{$-1$} (0,0) node[below]{ $0$} (2/3,0) node[below]{$\frac{1}{4v}$} (2,0) node[below]{$1$};
        \draw (2.8,3.5) node {(d)};
        \draw (1.3,3.5) node {(c)};
        \draw (-0.9,3.5) node {(b)};
        \draw (-3,3.5) node {(a)};
        
        \draw[gray!75,dashed] (2,0) -- (2,4);
        \draw[gray!75,dashed] (2/3,0) -- (2/3,4);
        \draw[gray!75,dashed] (-2,0) -- (-2,4);
        

        \draw[arrows = {-Stealth[inset=0pt, length=8pt, angle'=30]},ultra thick,gray!75, ] (2.3,2)--(1.3,2);
        
        \draw[arrows = {-Stealth[inset=0pt, length=8pt, angle'=30]},ultra thick,gray!75, ] (3.6,1)--(2.6,1);

        \draw[arrows = {-Stealth[inset=0pt, length=8pt, angle'=30]},ultra thick,gray!75,](1.2,0.5) -- (1.2,1.5);
        
        \draw[arrows = {-Stealth[inset=0pt, length=8pt, angle'=30]},ultra thick,gray!75, ] (-0.5,1)--(-1.5,1);

        \draw[arrows = {-Stealth[inset=0pt, length=8pt, angle'=30]},ultra thick,gray!75, ] (-1.7,2)--(-2.7,2);

        \draw[arrows = {-Stealth[inset=0pt, length=8pt, angle'=30]},ultra thick,gray!75, ] (2/3+0.3,2)--(-1/3+0.3,2);


        \node[above,font=\footnotesize\bfseries] at (0.3,5) {\boxed{\text{$v >  \tfrac14$}}};
    \end{tikzpicture}
    \caption{To prove that periodic orbits do not exist, we split of the phase-space of equations \eqref{LRSIXeq1}-\eqref{LRSIXeq2} by the isoclines $\Sigma_+=-1,1/(4v),+1$, where $\Sigma_+'<0$. In the regions (b) and (d), the flow is monotone due to $\Sigma_+'<0$. Similarly, in the region (c) monotonicity entails due to $N_1'>0$. In region (a), we use the Dulac-Bendixson criterion to exclude periodic orbits. 
    }\label{fig:isoclines}
\end{figure}
%
%

As a consequence of the absence of periodic orbits, the Poincaré-Bendixson theorem implies that $\omega$-limit sets are equilibria, heteroclinic cycles or homoclinics; we recall all solutions are bounded, due to Lemma~\ref{lem:bdd}. Note there are no heteroclinic cycles, nor homoclinics, due to the linear structure of fixed points and the known heteroclinics within the invariant set $N_1=0$. 
On one hand, the $\omega$-limit set of all solutions with $N_1>0$ consists of the equilibrium $(\Sigma_+, N_1)=(-1,0)$. 
On the other hand, the $\alpha$-limit set of all bounded solutions with $N_1>0$ consists of the following: 
when $v\in [1/4,1)$, the equilibrium $(\Sigma_+, N_1)=(1/(4v),0)$, except for the heteroclinic between $(\Sigma_+, N_1)=(1,0)$ and $(\Sigma_+, N_1)=(-1,0)$; 
when $v\in (0,1/4)$, the equilibrium $(\Sigma_+, N_1)=(1,0)$, except for the heteroclinic between $(\Sigma_+, N_1)=(1/(4v),0)$ and $(\Sigma_+, N_1)=(-1,0)$.
Note these heteroclinics exist, since in each of the aforementioned case, the equilibria $(\Sigma_+, N_1)=(1/(4v),0)$ and $(\Sigma_+, N_1)=(1,0)$ have a 1d unstable manifold with $N_1>0$ and the flow is dissipative such that all solutions converge to the equilibrium $(\Sigma_+, N_1)=(-1,0)$.
For $N_1 < 0$, analogous results hold due to the system's symmetry.

Lastly, we ascend from the system \eqref{LRSIXeq1}-\eqref{LRSIXeq2} to include the dynamics for the variable $N_2$ in \eqref{LRSIXeq3}. 
Note that $(\Sigma_+, N_1)= \left(1/(4v), 0\right)$ implies $N_2 = \infty$, due to \eqref{N2}, and thereby this point does not lie in the original constrained phase-space $\mathcal{LRS}_1 \subset \mathbb{R}^5$, but at infinity.
Similarly, the points with $(\Sigma_+, N_1)= \left(\pm 1, 0\right)$ and $N_2 = \infty$ correspond with the point at infinity of the lines $\mathrm{LRS}^\pm$.
For $v\in (1/2,1)$, none of these points attract 
a solutions from the $\mathcal{LRS}_1 \subset \mathbb{R}^5$ phase-space, due to dissipativity in Lemma~\ref{lem:bdd}. 
In any case, to discern the behavior of the variable $N_2$, note that the linearization of~\eqref{LRSIXeq3} at $\mathrm{LRS}^-$ amounts to $N_{2}' = 2(1-2v) N_{2}$, and thus in a sufficiently small neighborhood of $\mathrm{LRS}^-$, the variable $N_2$ increases for $v<1/2$, it remains constant for $v=1/2$ and it decreases for $v>1/2$.

For the backwards in time asymptotics (i.e. the $\alpha$-limit sets), 
it is not hard to see that all solutions of the two-dimensional flow given by \eqref{LRSIXeq1}-\eqref{LRSIXeq2} are contained in the negative-invariant set, $\Sigma_+>0$, for sufficiently large negative times (i.e. $\tau_-\ll 0$). More specifically, as $\tau_-\to -\infty$, we have $\Sigma_+\to 1,1/(4v)$ or $\infty$; which in turn respectively dictates that $N_1\to 0$ or $\infty$; and thus $N_2\to \infty$, due to~\eqref{N2}.
\qed
\end{pf}





\subsubsection{Generic Dynamics}

For the generic dynamics, we exclude possible invariant sets as $\omega$-limits, and thereby we conclude what are the actual possible asymptotic limits; see \cite[Lemmata 5.3,5.4]{heiugg09b}.
Note that for the case $v\in (1/2,1)$, the situation is simpler than $v=1/2$. Indeed, one does not have to exclude the point at infinity of each of the invariant lines $\mathrm{LRS}^\pm$, nor a point such that $\Sigma_+=-1/(2v)$ and $N_+=\infty$, since the system is dissipative due to Lemma~\ref{lem:bdd}. However, these points will play a major role in the dynamics for the case $v\in (0,1/2)$. 
%
\begin{lem}\label{lem:generic}
    Consider a non-$\mathcal{LRS}_\alpha,\alpha=1,2,3,$ type $\mathrm{VIII}$ or $\mathrm{IX}$ solution for $v\in (1/2,1)$. Then, its $\omega$-limit set does not contain a point of type $\mathrm{VII}_0$, including $\mathrm{LRS}^\pm$. 
\end{lem}
\begin{pf}
Suppose, towards a contradiction, that a non-$\mathcal{LRS}_\alpha,\alpha=1,2,3,$ type $\mathrm{VIII}$ or $\mathrm{IX}$ solution (which is bounded for $\tau_-\in\mathbb{R}_+$ due to Lemma~\ref{lem:bdd}) has an $\omega$-limit point which is of type $\mathrm{VII}_0$ (and thereby not of types $\mathrm{I}$ and $\mathrm{II}$). Due to invariance of the $\omega$-limit, the whole orbit of type $\mathrm{VII}_0$ that passes through this point is contained within such a $\omega$-limit set. 
However, Proposition~\ref{lem:genVIVIIalpha} implies that all solutions of type $\mathrm{VII}_0$ are unbounded as $\tau_- \to -\infty$, which contradicts the fact that the $\omega$-limit set is compact.
\qed
\end{pf}
The above argument shows that all type $\mathrm{IX}$ supercritical solutions converge to the union of Bianchi types $\mathrm{I}$ and $\mathrm{II}$ solutions. 
This is in contrast to GR, where the non-generic sets $\mathcal{LRS}_\alpha$ consist of solutions that converge to the $\mathrm{LRS}$ subset.

\section{Conclusion}\label{sec:conc}

In this work we have analyzed the global asymptotic (forward) dynamics of the Bianchi type~$\mathrm{IX}$ in certain modified gravity theories, given by equation~\eqref{intro_dynsyslambdaR}, with $v>1/2$. 
Using dynamical systems techniques, we have proved that \emph{all} solutions of Bianchi type~$\mathrm{IX}$ converge to an analogue of the Mixmaster attractor, consisting of Bianchi type I solutions (Kasner states) and heteroclinic chains of Bianchi type II solutions. In contrast to GR, generic solutions of these supercritical models converge to a unique stable Kasner state and thus do not display oscillations, whereas a set of measure zero oscillates. In particular, we emphasize that for $v>1/2$, there are no residual sets of small measure which are arbitrarily large (which is the case for the $\mathrm{LRS}$ solutions when $v=1/2$). This greatly simplifies the proof.

Additionally, we have detected a new bifurcation in the locally rotationally symmetric subsets, $\mathcal{LRS}_\alpha$, at the parameter value \(v = 1/4\). This amounts to a transcritical bifurcation at infinity, in which some object from outside the constrained phase-space invades it. 
This transition may also reflect in a reorganization of bounded objects and suggests a corresponding modification of the associated stable manifolds of certain objects (e.g., heteroclinic chains), which may influence the backward asymptotics of nearby trajectories. Although this phenomenon does not affect the supercritical attractor result proved here, it indicates an additional geometric mechanism governing the $\mathcal{LRS}_\alpha$ dynamics that is relevant for subcritical regime.
We emphasize that this threshold, \(v = 1/4\), does not play a role in the lower Bianchi types (e.g. $\mathrm{I}, \mathrm{II}, \mathrm{VI}_0, \mathrm{VII}_0$) and only appears as a parameter value where the intersection of disks where cross-terms grow coincides with the points $\mathrm{Q}_\alpha$ (cf. \cite[Fig. 18]{HellLappicyUggla}). 

For Bianchi type~$\mathrm{VIII}$, with $v>1/2$, the attractor conjecture remains open, as one still needs to exclude the fixed point \( p_{\mathrm{VI}_0} \) as an element of an $\omega$-limit set of a Bianchi type~$\mathrm{VIII}$ generic solution. We know this saddle-point possess a 1d strong-stable manifold, $W^{ss}(p_{\mathrm{VI}})$, and thereby only attract one type~$\mathrm{VIII}$ solution (which is a non-generic subset); see Lemma~\ref{lem:strong-stable-pVI}. 
%
%
%
Numerical analysis suggests that this strong-stable manifold is unbounded (backwards in time); see Figure~\ref{fig:Wss}. 
If this is indeed the case, we would be able to exclude such stable manifold of the global attractor using a similar argument as in Lemma~\ref{lem:generic}. However, we cannot rigorously rule out this behavior with the current analysis, as we would need to rigorously locate such stable manifold.

Moreover, for Bianchi type~$\mathrm{VIII}$ with $v>1/2$, it is still not clear what happens to $W^{ss}(p_{\mathrm{VI}})$ as $v\to 1/2$, besides the fact that $p_{\mathrm{VI}} \to \mathrm{T}_1$. In this limit, some of the coefficients in the expansion described in Lemma~\ref{lem:strong-stable-pVI} vanish, whereas some diverge to infinity. This indicates that such strong-stable manifold does not degenerate only at the Taub point $\mathrm{T}_1$, but in a possibly bigger subset. Numerical explorations suggest that it degenerates at three different sets with different behaviors: first, it follows the type $\mathrm{II}$ heteroclinic orbit to $\mathrm{T}_1$ from $\mathrm{Q}_1$, 
second, it follows a type $\mathrm{II}$ orbit with a unique bump $N_3<<0$,
third, it has a last transient oscillatory behavior until it tracks some trajectory in the $\mathcal{LRS}_\alpha$ subset with growing $N_1 \sim N_2$ and reaches limiting $\alpha$-limit set of the strong-stable manifold; which is the intersection of two disks where cross-terms grow. See Figure~\ref{fig:WssLIMIT}. 

%
\begin{figure}[H]
    \centering
    \includegraphics[scale=0.36]{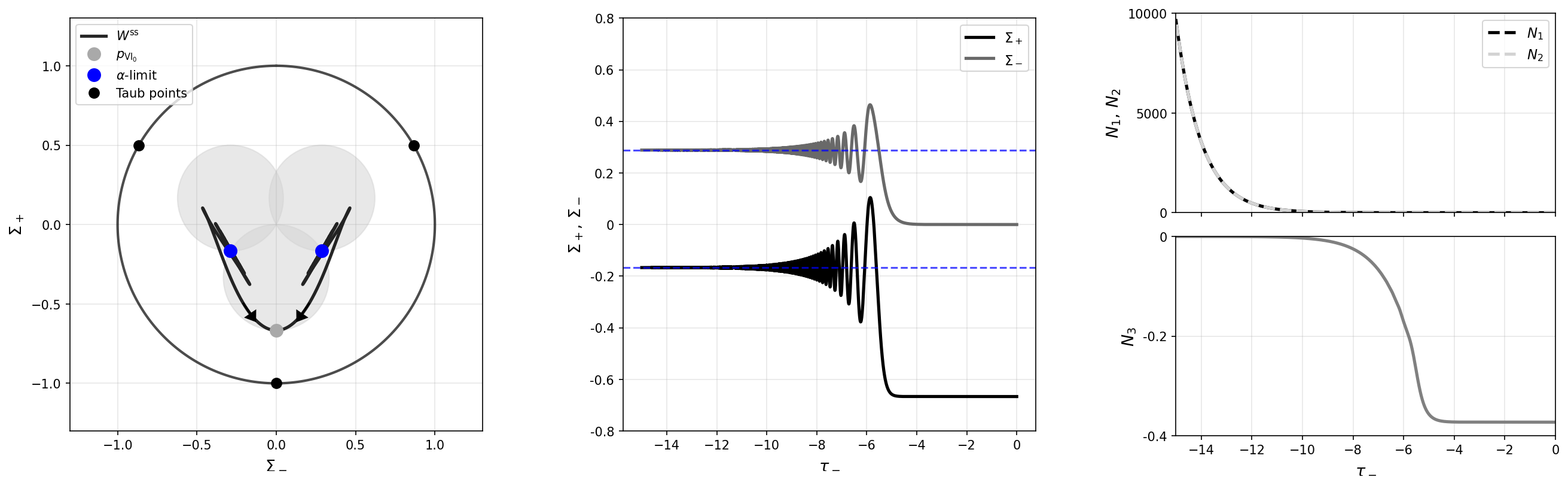}
    \caption{Numerical stable manifold $W^{\mathrm{ss}}(p_{\mathrm{VI}_0})$ for $v = 0.75$ from Lemma~\ref{lem:strong-stable-pVI}, which consists of 
    two symmetric solutions with initial data satisfying \eqref{Wsspviv3} with $\Sigma_-=10^{-10}$ and $\Sigma_-= - 10^{-10}$. 
    \textbf{Left:} Projection of the solutions into the $\Sigma$-plane containing the Kasner circle, the fixed point $p_{\mathrm{VI}_0}$ (gray dot), and the numerical $\alpha$-limit sets (blue dots) of the set $W^{\mathrm{ss}}(p_{\mathrm{VI}_0})$. 
    Here, the light gray disks represent regions where the cross-terms $N_\alpha N_\beta$ grow, see \cite[Eq. 109, Fig. 18]{HellLappicyUggla}; in particular, $p_{\mathrm{VI}_0}$ lies in the extremum of a disk, whereas the $\alpha$-limit sets 
    lie in the intersection of two disks.
    \textbf{Middle and right:} The backwards evolution of the solution with initial data $\Sigma_-=10^{-10}$ with variables $\Sigma_-,\Sigma_+$ (which are bounded) and $N_1,N_2,N_3$ ($N_3$ goes to 0 and $N_1\approx N_2$ grow).
    }
    \label{fig:Wss}
\end{figure}

\begin{figure}[H]
    \centering
    \includegraphics[scale=0.36]{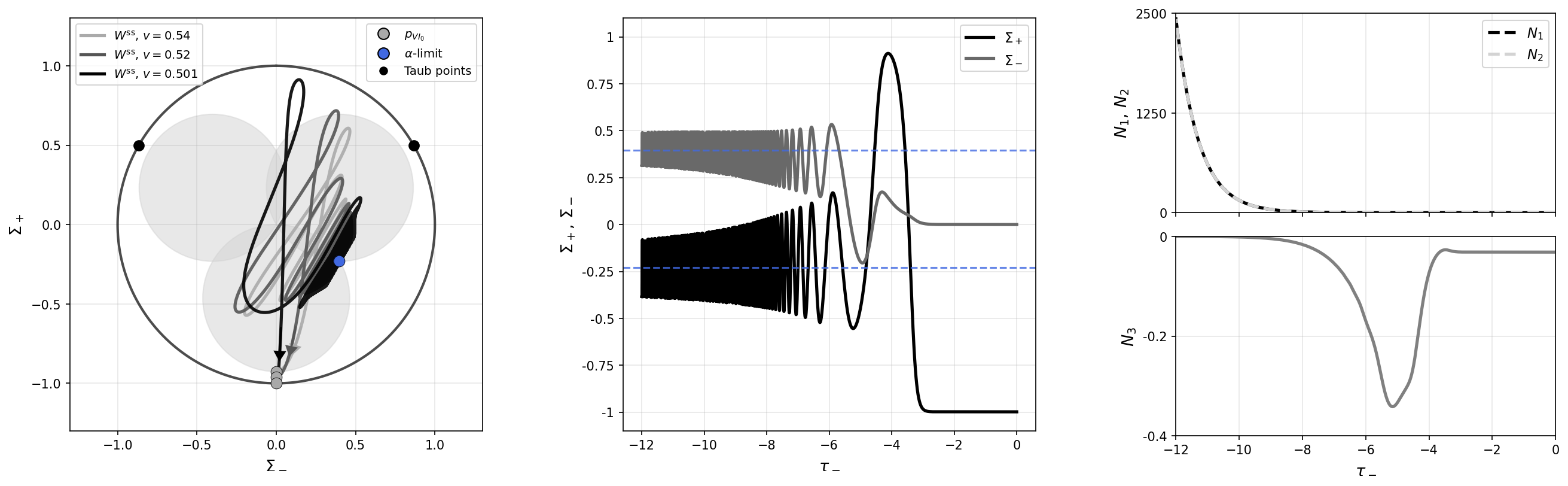}
    \caption{Numerical stable manifold $W^{\mathrm{ss}}(p_{\mathrm{VI}_0})$ as $v \to 1/2$. 
    \textbf{Left:} Projection of the solutions into the $\Sigma$-plane. 
    \textbf{Middle and right:} For $v=0.501$, the backwards evolution of the solution with initial data $\Sigma_-=10^{-10}$ and variables $\Sigma_-,\Sigma_+$ (which are bounded), $N_3$ (which tracks a type~$\mathrm{II}$ or $\mathrm{VI}_0$ trajectory and then goes to 0), and $N_1,N_2$ (which grow following a LRS solution, $N_1\approx N_2$).
    }
    \label{fig:WssLIMIT}
\end{figure}

Most notably, the conjecture attractor remains open in the subcritical regime, \( v<1/2 \), as it exhibits a fundamental different behavior. For this range of parameters, the Taub points become unstable and both invariant subsets $\mathrm{LRS}^\pm$ and $\mathcal{LRS}_\alpha$ consist of blow-up solutions; we thus raise the question if there is a \emph{compact} global attractor. 
In particular, it is currently unclear whether generic solutions converge to a Mixmaster--type attractor (consisting solely of heteroclinics of types $\mathrm{I}, \mathrm{II}$), or more surprisingly, if generic solutions exhibit unbounded growth (and thus there is an \emph{unbounded} global attractor). 
We hope that answering these questions will shine a light in the attractor conjecture for Bianchi type  $\mathrm{VI}_{-1/9}$, see \cite{LappicyUggla}. 





\textbf{Acknowledgments.} 
E. Beatriz was supported by FAPESP, 2023/07941-7 and 2024/05580-0.			
E. Bonotto was supported by FAPESP, 2024/07880-0 and 2020/14075-6, and CNPq, 316169/2023-4.			
P. Lappicy was supported by Marie Skłodowska-Curie Actions, H2020 Cofund, 847635, UNA4CAREER, project DYNCOSMOS. 

\bibliographystyle{plain}

\begin{thebibliography}{90}

\bibitem{Bakas10} I.~Bakas, F.~Bourliot, D.~Lüst, and M.~Petropoulos.
\newblock The {M}ixmaster universe in Hořava–Lifshitz gravity.
\newblock {\it Class. Quantum Grav. } {\bf 27}, 045013, (2010).

\bibitem{Beguin} F.~Béguin.
\newblock Aperiodic oscillatory asymptotic behavior for some Bianchi spacetimes.
\newblock {\it Class.\ Quant.\ Grav.} \textbf{27}, 185005, (2010).

\bibitem{Dutilleul} F. Béguin and T. Dutilleul. 
\newblock Chaotic Dynamics of Spatially Homogeneous Spacetimes. 
\newblock {\it Commun. Math. Phys.} {\bf 399}, 737-927, (2023). 

\bibitem{bkl70} V.~A.~Belinski\v{\i}, I.~M.~Khalatnikov, and E.~M.~Lifshitz.
\newblock Oscillatory approach to a singular point
in the relativistic cosmology.
\newblock {\it Adv.\ Phys.} {\bf 19}, 525, (1970).

\bibitem{bkl82} V.~A.~Belinski\v{\i}, I.~M.~Khalatnikov, and E.~M.~Lifshitz.
\newblock A general solution of the Einstein equations with a time singularity.
\newblock {\it Adv.\ Phys.} {\bf 31}, 639, (1982).

\bibitem{Bernhard}B.~Brehm.
\newblock Bianchi VIII and IX vacuum cosmologies: Almost every solution forms particle horizons and converges to the Mixmaster attractor. \href{https://arxiv.org/abs/1606.08058}{arXiv:1606.08058}, (2016).

\bibitem{Chernoff83} D. F.~Chernoff and J. D.~Barrow.
\newblock Chaos in the Mixmaster Universe.
\newblock {\it Phys. Rev. Lett.} {\bf 50}, 134, (1983).

\bibitem{LappicyLessard} K. E. M. Church, O. Hénot, P. Lappicy, J.-P. Lessard, and H. Sprink.
Periodic orbits in Hořava-Lifshitz cosmologies, \newblock {\it Gen. Rel. Grav.} {\bf 55},  2, (2023).

\bibitem{MG} T.~Clifton, P. G. Ferreira, A. Padilla, and C. Skordis.
\newblock Modified gravity and cosmology.
\newblock {\it Phys.\ Reports} {\bf 513}(1-3), 1-189, (2012).


\bibitem{Coley} A. A.~Coley and G. F. R. Ellis.
\newblock Theoretical cosmology.
\newblock {\it Class.\ Quant.\ Grav.} {\bf 37}, 013001, (2020).

\bibitem{fR} A.~de Felice and S.~Tsujikawa.
\newblock $f(R)$ Theories.
\newblock {\it Living\ Rev.\ Rel.} {\bf 13}, 3, (2010).

\bibitem{Gao} F.~Gao and J.~Llibre.
\newblock Global dynamics of Hořava–Lifshitz cosmology with non-zero curvature and a wide range of potentials.
\newblock {\it The European Phys. J. C} {\bf 80}, 137, (2020).

\bibitem{Kamenshchik} L.~Giani and A.~Y.~Kamenshchik.
\newblock Ho{\v{r}}ava{\textendash}Lifshitz gravity inspired Bianchi-{II} cosmology and the mixmaster universe.
\newblock  {\it Class.\ Quant.\ Grav.} {\bf 34}, 085007, (2017).

\bibitem{heiugg09b} J.~M.~Heinzle and C.~Uggla.
\newblock A new proof of the Bianchi type IX attractor theorem.
\newblock {\it Class.\ Quant.\ Grav.} {\bf 26}, 075015, (2009).

\bibitem{Mixmaster}
J.~M.~Heinzle and C.~Uggla.
\newblock Mixmaster: Fact and Belief.
\newblock {\it Class.\ Quant.\ Grav.} \textbf{26}, 075016, (2009).

\bibitem{HellLappicyUggla} J.~Hell, P.~Lappicy, and C.~Uggla.
Bifurcations and Chaos in Ho{\v{r}}ava-Lifshitz Cosmology, {\it Adv. Theor. Math. Phys.} {\bf 26} 7, 2095–2211, (2022).


\bibitem{hor09a} P.~Ho\v{r}ava.
\newblock Membranes at Quantum Criticality.
\newblock {\it J.\ High Energy Phys.} {\bf 0903}, 020, (2009).

\bibitem{hor09b} P.~Ho\v{r}ava.
\newblock Quantum Gravity at a Lifshitz Point.
\newblock {\it Phys.\ Rev.\ D} {\bf 79}, 084008, (2009).


\bibitem{LappicyDaniel} 
\newblock P. Lappicy and V. H. Daniel.
\newblock Chaos in spatially homogeneous Hořava-Lifshitz subcritical cosmologies.
\newblock {\it Class. Quant. Grav.} \textbf{39} 13, 135017, (2022).

\bibitem{LappicyUggla} P.~Lappicy and C.~Uggla.
Oscillatory spacelike singularities: The Bianchi type $\mathrm{VI}_{-1/9}$ vacuum models, {\it https://doi.org/10.48550/arXiv.2410.10375} (2024). 

\bibitem{Liebscher}
S.~Liebscher, J.~Harterich, K.~Webster, and M.~Georgi.
\newblock Ancient dynamics in Bianchi models: Approach to Periodic Cycles.
\newblock
{\it Commun.\ Math.\ Phys.} \textbf{305}, 59-83, (2011).

\bibitem{lieetal12} S.~Liebscher, A.~D.~Rendall, and S.~B.~Tchapnda.
\newblock Oscillatory singularities in Bianchi models with magnetic fields.
\newblock {\it Ann. Henri Poincar\'e} {\bf 14}, 1043--1075, (2013).


\bibitem{mis69a} C.~W.~Misner.
\newblock Mixmaster universe.
\newblock {\it Phys.\ Rev.\ Lett.} {\bf 22}, 1071, (1969).



\bibitem{Miso11} Y.~Misonoh, K.~Maeda, and T.~Kobayashi.
\newblock Oscillating Bianchi IX universe in Hořava-Lifshitz gravity.
\newblock {\it Phys. Rev. D} {\bf 84}, 064030, (2011).

\bibitem{Ringstrom2} H.~Ringström.
\newblock Curvature blow up in Bianchi VIII and IX vacuum spacetimes.
\newblock {\it Class.\ Quant.\ Grav.} \textbf{17}, 713 (2000).

\bibitem{Ringstrom} H.~Ringström.
\newblock The Bianchi IX attractor.
\newblock {\it Annales\ Henri\ Poincaré} \textbf{2}, 405-500 (2001).

\bibitem{HL_status_report} T.~P.~Sotiriou.
\newblock Ho\v{r}ava-Lifshitz gravity: a status report.
\newblock{\it J. Phys. Conf. Ser.} {\bf 283}, 012034, (2011).

\bibitem{fR2} T. P.~Sotiriou and V.~Faraoni.
\newblock $f(R)$ theories of gravity.
\newblock {\it Rev.\ Mod.\ Phys.} {\bf 82}(1), 451-497, (2010).

\bibitem{Tak94} F. Takens.
\newblock Heteroclinic attractors: Time averages and moduli of topological conjugacy.
\newblock {\it Bull.\ Br.\ Math.\ Soc.} {\bf 25}, 107-120, (1994).

\bibitem{LRS}
C.~Uggla and H. Zur Muhlen.
\newblock Compactified and reduced dynamics for locally rotationally symmetric Bianchi type IX perfect fluid models.
\newblock {\it Class.\ Quant.\ Grav.} \textbf{7}, 1365, (1990).

\bibitem{ugg13a} C. Uggla.
\newblock Recent developments concerning generic spacelike singularities.
\newblock {\it Gen.\ Rel.\ Grav.} {\bf 45}, 1669, (2013).

\bibitem{ugg13b} C.~Uggla.
\newblock Spacetime Singularities: Recent Developments.
\newblock {\it Int.\ J.\ Mod.\ Phys.\ D} {\bf 22}, 1330002, (2013).

\end{thebibliography}

\end{document}